\documentclass[preprint]{elsarticle}

\usepackage{amssymb}
\usepackage{amsthm}
\usepackage{amsfonts}
\usepackage{amsmath}
\usepackage{color}
\usepackage{natbib}
\usepackage{bbm}            
\usepackage{color}
\usepackage{url}
\usepackage{tikz-cd}
\usepackage{bm}
\usepackage{algpseudocode}
\usepackage{algorithm}
\usepackage{todonotes}

\newcommand{\bN} { {\mathbb{N}}}

\newcommand{\bQ} { {\mathbb{Q}}}
\newcommand{\bZ} { {\mathbb{Z}}}

\newcommand{\cM} { {\mathcal{M}}}

\def\bm#1{\mathchoice{\kern-.5pt\mathord{\text{\bfseries\itshape#1}}\kern+.25pt}
                      {\kern-.5pt\mathord{\text{\bfseries\itshape#1}}\kern+.25pt}
                      {\kern+.25pt\mathord{\text{\scriptsize\bfseries\itshape#1}}\kern+.5pt}
                      {\kern+.25pt\mathord{\text{\tiny\bfseries\itshape#1}}}\kern+.25pt}

\newcommand{\De} { {\Delta}}
\newcommand{\si} { {\sigma}}

\newcommand{\vv}{\mathbf{v}}
\newcommand{\vw}{\mathbf{w}}
\newcommand{\e}{\mathbf{e}}
\newcommand{\x}{\mathbf{x}}
\newcommand{\xh}{\widehat{\mathbf{x}}}
\newcommand{\bsi}{\boldsymbol{\si}}
\newcommand{\btau}{\boldsymbol{\tau}}

\newtheorem{theorem}{Theorem}

\newtheorem{corollary}[theorem]{Corollary}
\newtheorem{lemma}[theorem]{Lemma}
\newtheorem{remark}[theorem]{Remark}
\newtheorem{definition}[theorem]{Definition}
\newtheorem{example}[theorem]{Example}

\DeclareMathOperator{\Span}{Span}
\DeclareMathOperator{\res}{res}
\newtheorem{fact}[theorem]{Fact}
\newtheorem{problem}[theorem]{Problem}
\journal{xxxxxx}

\begin{document}

\begin{frontmatter}



\title{How to generate all possible rational \\Wilf--Zeilberger forms?
\footnote{S. Chen was partially supported by the National Key R\&D Program of China (No. 2020YFA0712300 and No. 2023YFA1009401), the NSFC grant (No. 12271511), CAS Project for Young Scientists in Basic Research (Grant No. YSBR-034), and the CAS Fund of the Youth Innovation Promotion Association (No. Y2022001).
C.\ Koutschan and Y.\ Wang were partially supported by the Austrian Science Fund (FWF): 10.55776/I6130. All authors also were supported by the International Partnership Program of Chinese Academy of Sciences (Grant No.\ 167GJHZ2023001FN).}
}

\author[a,b]{Shaoshi Chen}
\address[a]{%
  KLMM, Academy of Mathematics and Systems Science, \\
  Chinese Academy of Sciences, Beijing 100190, China}
\address[b]{%
  School of Mathematical Sciences, \\
  University of Chinese Academy of Sciences, Beijing 100049, China}
\ead{schen@amss.ac.cn}

\author[c]{Christoph Koutschan}
\address[c]{%
	Johann Radon Institute for Computational and Applied Mathematics (RICAM),\\
  Austrian Academy of Sciences, Altenberger Stra\ss e 69, 4040 Linz, Austria}
\ead{christoph.koutschan@ricam.oeaw.ac.at}

\author[a,b,c]{Yisen Wang}
\ead{wangyisen@amss.ac.cn}

\begin{abstract}
  Wilf--Zeilberger pairs are fundamental in the algorithmic theory of Wilf and
  Zeilberger for computer-generated proofs of combinatorial
  identities. Wilf--Zeilberger forms are their high-dimensional
  generalizations, which can be used for proving and discovering convergence
  acceleration formulas. This paper presents a structural description of all
  possible rational such forms, which can be viewed as an additive
  analog of the classical Ore--Sato theorem. Based on this analog, we show a
  structural decomposition of so-called multivariate hyperarithmetic expressions,
  which extend multivariate hypergeometric terms to the additive setting.
\end{abstract}


\end{frontmatter}


\section{Introduction}\label{sec:intro}

In the 1990s, Wilf and Zeilberger developed an algorithmic theory for proving combinatorial identities~\cite{Wilf1992, PWZbook1996}.
The notion of WZ-pairs is one of the core concepts in their theory, which was originally introduced in~\cite{WilfZeilberger1992} with a recent brief description in~\cite{Tefera2010}. It is an elegant and powerful tool for proving identities involving sums of hypergeometric terms in an algorithmic fashion, and
there are implementations in several computer algebra systems.
A WZ-pair is a pair of functions~$\bigl(F(n, k), G(n, k)\bigr)$ satisfying the relation
\[ F(n+1, k) - F(n, k) = G(n, k+1) - G(n, k),
\]
where both $F$ and $G$ are hypergeometric terms, i.e., their shift quotients with respect to $n$ and $k$ are rational functions in $n$ and~$k$.
Assume that $F(n,k)$ vanishes except for $k$ in some finite interval for each~$n$. Therefore, one can sum both sides of the above equation w.r.t.~$k$ from $0$ to $\infty$ to get, by telescoping,
\[\sum_{k=0}^{\infty} F(n+1, k) - \sum_{k=0}^\infty F(n, k)= \lim_{k\rightarrow\infty} G(n, k+1) - G(n, 0).\]
If the boundary terms on the right-hand side vanish, we obtain that
\[\sum_{k=0}^{\infty} F(n+1, k) = \sum_{k=0}^\infty F(n, k),\]
which implies that the definite sum~$\sum_{k=0}^\infty F(n, k)$ is independent of~$n$. Thus, we get
the identity  $\sum_{k=0}^\infty F(n, k) = c$, where the constant $c$ can be determined by evaluating the sum for one specific value of~$n$. We can also get a companion identity by summing w.r.t.~$n$. For example, the WZ-pair 
$(F, G)$ with
\[F = \frac{\binom{n}{k}^2}{\binom{2n}{n}} \quad \text{and} \quad  G= {\frac {( 2k-3n-3) {k}^{2}}{ 2(2n+1)( -n-1+k)^{2}}} \cdot \frac{\binom{n}{k}^2}{\binom{2n}{n}}
 \]
leads to two identities
\[\sum_{k=0}^\infty \binom{n}{k}^2 = \binom{2n}{n} \quad \text{and} \quad \sum_{n=0}^\infty \frac{(3n-2k+1)}{(2n+1)\binom{2n}{n}} \binom{n}{k}^2 = 2. \]

WZ-pairs have been employed by Guillera to prove Ramanujan-type series~\cite{Guillera2002,Guillera2006,Guillera2010,Guillera2013} and new congruences involving harmonic numbers and Ap\'ery numbers conjectured by Sun~\cite{Sun2023,XiaSun2023,HouSun2023}. This shows that WZ-pairs are not only beneficial in combinatorics but also in the fields of mathematical analysis and number theory.
Wilf--Zeilberger forms, in short WZ-forms, are a direct generalization of WZ-pairs to tuples with more than two entries. The naming ``WZ-form'' is reminiscent of the classical concepts of differential forms~\cite{Cartan1970} and difference forms, to which it is indeed related. To be more precise, we first recall some basic terminologies and properties about difference forms from~\cite{Zeilberger1993,Mansfield2008}. Let $\cM$ be a well-chosen module of discrete functions defined on some region of $\bZ^n$ so that one can define the usual shift operators $\si_1, \ldots, \si_n$ that commute.  Let $\delta x_1,\ldots, \delta x_n$ be ``indeterminates" satisfying the relations
\[\delta x_i \, \delta x_j = - \delta x_j \, \delta x_i\quad \text{for all $i, j$ with $1\leq i \leq j\leq n$}.\] 
An (exterior) \emph{difference $k$-form} is a linear combination of~\lq\lq words\rq\rq~in the alphabet $\{\delta x_1,\ldots, \delta x_n\}$ with coefficients from the module $\cM$, which can be written as
\[
  \omega = \sum_{i_1, \ldots, i_k} f_{i_1, \ldots, i_k} \,  \delta x_{i_1} \cdots \delta x_{i_k}, \quad \text{where $f_{i_1, \ldots, i_k}\in \cM$}.
\]
Any element of $\cM$ can be seen as a~\emph{$0$-form}.
The \emph{exterior difference} $\delta$ of~$ \omega$ is the $(k+1)$-form defined by
\[
  \delta \omega = \sum_{i_1, \ldots, i_k} \left(\sum_{i=1}^n \biggl(\si_i(f_{i_1, \ldots, i_k})- f_{i_1, \ldots, i_k}\biggr)\delta_{x_i} \right) \,  \delta x_{i_1} \cdots \delta x_{i_k}
\]
A difference form $\omega$ is called \emph{closed} if~$\delta \omega =0 $ and it is called \emph{exact} if  $\omega = \delta \theta$ for some~$\theta$.
Analogous to the de Rham complex for differential forms, the property $\delta^2 = 0$ holds for difference forms since the shift operators commute. That is to say, every exact form is always closed, but the converse is not true in general. Closed difference $1$-forms with hypergeometric coefficients are called \emph{WZ-forms} in~\cite{Zeilberger1993} and~\emph{WZ-cohomology} is the quotient of the modules of closed forms modulo that of exact forms. Similar to WZ-pairs, WZ-forms as well can be used to prove combinatorial identities and to derive convergence acceleration formulas~\cite{Zimmermann2000}. For example, Dixon's identity
\begin{equation}\label{eq:Dixon}
  \sum_k (-1)^k \frac{(a+b)!(a+c)!(b+c)!a!b!c!}{(a+k)!(a-k)!(b+k)!(b-k)!(c+k)!(c-k)!(a+b+c)!} = 1
\end{equation}
can be derived from the closed difference $1$-form
\[
	\omega = F \, \delta k + G \, \delta a + H \, \delta b + I \, \delta c,
\]
where $F$ denotes the summand in~\eqref{eq:Dixon}, and
\begin{align*}
	G := & - \frac{(b+k)(c+k)}{2(a-k+1)(a+b+c+1)} \cdot F,\\
	H := & - \frac{(a+k)(c+k)}{2(b-k+1)(a+b+c+1)}\cdot F,\\
	I := & - \frac{(a+k)(b+k)}{2(c-k+1)(a+b+c+1)} \cdot F.
\end{align*}
The idea is similar to that used in deriving identities from WZ-pairs.  For any given closed difference $1$-form $\omega = f_1 \, \delta x_1 + \cdots + f_n \, \delta x_n$, 
Zeilberger applied the discrete Stokes theorem~\cite[p.\ 583]{Zeilberger1993} under some finitely supported conditions on the $f_i$'s to prove that $g(\widehat{\mathbf{x}}_i) = \sum f_i \, \delta x_i$ is identically constant in~$\widehat{\mathbf{x}}_i=(x_1,\dots,x_{i-1},x_{i+1},\dots,x_n)$ and this constant can be determined by checking the sum for one special value of $\widehat{\mathbf{x}}_i$.  To get convergence acceleration formulas, the idea is to apply the discrete Stokes theorem to a closed difference $1$-form  $\omega = F(n,k) \, \delta k + G(n,k) \, \delta n $ to obtain
	\[
	\sum_{n=0}^{\infty} G(n,0) = \sum_{n=1}^{\infty} \biggl( F(n,n-1)+G(n-1,n-1)\biggr)
	\]
	whenever both sums converge. For example, the closed form $\omega = F \, \delta k + G \, \delta n$ with
\begin{equation*}
     F = \frac{(-1)^{n+k}k!^2(n-k-1)!}{(n+k+1)!}\quad  \text{and}\quad 
     G = \frac{2\,(-1)^{n+k}k!^2(n-k)!}{(n+1)(n+k+1)!}.
\end{equation*}
leads to the convergence acceleration formula~\cite{Poorten1978} for~$\zeta(2)$
\begin{equation}\label{eq:zeta2}
  \zeta(2) = 2 \sum_{n=0}^{\infty} \frac{(-1)^n}{(n+1)^2} =
  3 \sum_{n=1}^{\infty} \frac{1}{\binom{2n}{n}n^2}.
\end{equation}
Then the natural questions arise how one can construct such WZ-forms? And how can we decide whether a WZ-form is exact or not? Note that the idea of deriving WZ-forms from known identities and then employing them to generate new identities has been shown in~\cite{Gessel1995,Zimmermann2000,Mu2019,Au2025}. 
In this paper, we restrict our attention to WZ-forms with rational functions instead of hypergeometric entries and use $n$-tuples $(f_1, \ldots, f_n)$ to denote WZ-forms instead of difference forms. We shall describe the structure of rational WZ-forms which is an additive analog of the Ore--Sato theorem~\cite{Ore1930, Sato1990, Gelfand, AbramovPetkovsek2002a, Abramov2008} in Theorem~\ref{thm:addOreSato}. Before proving the main theorem, we first recall the structure theorem on WZ-pairs from~\cite{Chen2019} in Section~\ref{sec:wzpair} which will be used as the base case in our induction proof and then overview some basic properties about orbital decomposition and orbital residues in Section~\ref{sec:orbitaldecomp}. The proof of Theorem~\ref{thm:addOreSato} splits into two steps: the first step is to show that any WZ-form can be decomposed into one exact WZ-form plus several uniform WZ-forms in Section~\ref{sec:decompofWZforms} and the second step is to describe the explicit integer-linear structure of 
uniform WZ-forms in Section~\ref{sec:uniformWZ-forms}. In the last section, we 
present an algorithm for computing the additive structure of WZ-forms that minimizes the uniform part. With the help of this minimal decomposition, we can detect the exactness of WZ-forms by just checking whether the uniform part vanishes.

\section{Preliminaries}

Throughout this paper, let $\bN$ denote the set of nonnegative integers. Let $ K $ be an algebraically closed field of characteristic zero and $ K(x_{1},\dots,x_{n}) $ be the field of rational functions in the variables $ x_1,\ldots,x_n $ over $ K $, which is also written as $ K(\mathbf{x}) $. 
For a multivariate function~$f$, the shift maps~$\si_i$ are defined as 
\[
  \sigma_i f(x_1,\ldots,x_n) = f(x_1,\ldots,x_i+1,\ldots,x_n),
  \quad \forall i \in \{1,\ldots,n\}.
\]
The action of operators on functions is also denoted by~$\bullet$, e.g.,
$\si_i\bullet f = \si_i(f)$.
Analogously, the forward difference operators are defined as
\[
	 \De_{i}(f):=\si_{i}(f)-f,
	 \quad\forall i \in \{1,\dots,n\}.
\]

\begin{definition}[Hypergeometric, hyperarithmetic]\label{def:hypg,hypa}
  A nonzero expression~$H$ is said to be \emph{hypergeometric} over
  $K(\mathbf{x})$ if there exist rational functions
  $ f_1,\dots,f_n \in K(\mathbf{x})$ such that
  \[
    \dfrac{\si_{i} (H)}{H}=f_i , \quad\forall i \in \{ 1, \dots, n\}.
  \]
  Analogously, $H$ is said to be \emph{hyperarithmetic} over
  $K(\mathbf{x})$ if there exist rational functions
  $ f_1,\dots,f_n \in K(\mathbf{x})$ such that
  \[
    \si_{i} (H)-H=f_i, \quad\forall i \in \{ 1, \dots, n\}.
  \]
  In both cases, the rational functions $f_1,\dots,f_n$ are called
  the \emph{certificates} of~$H$.
  Two hypergeometric (resp.\ hyperarithmetic) expressions~$ H_1 $ and~$ H_2 $
  are called \emph{conjugate}, denoted by $ H_1\simeq H_2 $, if they have the
  same certificates.
\end{definition}
Since $ \si_i $ and $\si_j $ commute,
the certificates $ f_1,\dots,f_n $ of a hypergeometric term~$ H $ satisfy the following compatibility conditions:
\begin{equation}\label{eq:compatiblehypergeometric}
 \frac{\si_{i}(f_j)}{f_j}=\frac{\si_j(f_i)}{f_i},
 \quad\forall i,j \in \{ 1,\dots, n\}.
\end{equation}

Similarly, the certificates $ f_1,\dots,f_n $ of a hyperarithmetic expression~$H$
satisfy the following compatibility conditions:
\begin{equation}\label{eq:compatiblehyperarithmetic}
	\si_{i}(f_j)-{f_j}=\si_j(f_i)-{f_i} ,\quad\forall i,j \in \{ 1,\dots, n\}.
\end{equation}
\begin{definition}
	An $n$-tuple $ (f_1,\ldots,f_n) \in K(\mathbf{x})^{n}$ is called a \emph{WZ-form} with respect to $ (\De_1,\ldots,\De_n) $ if $ \De_i(f_j)=\De_j(f_i)$ for all $i,j\in\{ 1,\dots,n\}$. 
	Note that $(f_1,\ldots,f_n)$ is a WZ-form with respect to $ (\De_1,\ldots,\De_n) $ if and only if $\omega = f_1 \delta_1 + \cdots + f_n \delta_n$ is a closed difference $1$-form. 
\end{definition}

The classical Ore--Sato theorem plays an important role in the theory of multivariate hypergeometric terms~\cite{Gelfand, AbramovPetkovsek2002a, Abramov2008}, because it describes the multiplicative structure of nonzero rational functions $ f_1,\dots,f_n \in K(\mathbf{x}) $ that satisfy the compatibility conditions~\eqref{eq:compatiblehypergeometric}. The bivariate case was proven by Ore~\cite{Ore1930} and the multivariate case by Sato~\cite{Sato1990}.  According to this theorem, any multivariate hypergeometric term can be decomposed into a product of a rational function and several factorial terms (which are basically products of Gamma functions).

\begin{theorem}[Ore--Sato theorem]\label{thm:OreSato}
	Let $ f_1,\ldots,f_n \in K(\mathbf{x})$ be nonzero rational functions satisfying the compatibility conditions~\eqref{eq:compatiblehypergeometric}. Then there exist a rational function $ a\in K(\mathbf{x}) $, constants $ \mu_1,\ldots,\mu_n \in K$, a finite set $ V\subset \bZ^n $, and for each $ \mathbf{v}\in V $ a univariate monic rational function $ r_{\mathbf{v}}\in K(z) $ such that
	\[f_j=\frac{\si_{j}(a)}{a}\mu_{j}\prod_{\mathbf{v}\in V}\sideset{}{_0^{v_j}}
	\prod_{\ell} r_{\mathbf{v}}(\mathbf{v}\cdot \mathbf{x}+\ell),
    \]
	where $ \mathbf{v} \cdot \mathbf{x} := v_1x_1+\dots+v_nx_n $ and where the product notation is defined as follows: for~$s, t\in \bZ$,
\begin{equation*}
	\sideset{}{_s^t}\prod_\ell \alpha_\ell:=
	\begin{cases}
        \alpha_s\alpha_{s+1}\cdots\alpha_{t-1},
        & \text{if } t\geq s; \\[0.5em]
        \dfrac{1}{\alpha_t\alpha_{t+1}\cdots\alpha_{s-1}},
        & \text{if } t<s.
	\end{cases}
\end{equation*}
\end{theorem}
Christopher's theorem~\cite{Christopher1999,Zoladek1998} is an analog of the Ore--Sato theorem in the continuous case. Other analogs concern the $ q $-discrete case~\cite{DuLi2019} and the continuous-discrete case~\cite{Chen2011}. In this paper, we want to explore the additive structure of nonzero rational functions $ f_1,\dots,f_n \in K(\mathbf{x}) $ satisfying the compatibility conditions~\eqref{eq:compatiblehyperarithmetic}, i.e., $(f_1,\dots,f_n)$ is a WZ-form. Our main result, which is stated in the following theorem, reveals this additive structure and therefore implies an additive decomposition of hyperarithmetic expressions.

\begin{theorem}[Additive Ore--Sato theorem]\label{thm:addOreSato}
	Let $ f_1,\ldots,f_n \in K(\mathbf{x})$ be nonzero rational functions satisfying the compatibility conditions~\eqref{eq:compatiblehyperarithmetic}. Then there exist a rational function $ a\in K(\mathbf{x}) $, constants $ \mu_1,\ldots,\mu_n \in K$, a finite set $ V\subset \bZ^n $, and for each $ \mathbf{v}\in V $ a univariate monic rational function $ r_\mathbf{v}\in K(z) $ such that
	\[
    f_j=\si_{j}(a)-a+\mu_{j}+\sum_{\mathbf{v}\in V}\sideset{}{_0^{v_j}}
	\sum_{\ell}r_{\mathbf{v}}(\mathbf{v} \cdot \mathbf{x}+\ell),
    \]
	where $ \mathbf{v} \cdot \mathbf{x} := v_1x_1+\dots+v_nx_n $ and where we use the sum notation (for~$s, t\in \bZ$)
    \[
	\sideset{}{_s^t}\sum_\ell \alpha_\ell:=
	\begin{cases}
		 \alpha_s+\alpha_{s+1}+\dots+\alpha_{t-1}, & \text{if } t\geq s; \\[0.5em]
		-(\alpha_t+\alpha_{t+1}+\dots+\alpha_{s-1}), & \text{if } t<s.
	\end{cases}
    \]
\end{theorem}
In the proof of the classical Ore--Sato theorem, the complete irreducible factorization was
used as a key ingredient. When it comes to the additive case,
we need another auxiliary tool, the so-called orbital decomposition, which compensates the missing of partial fraction decompositions of multivariate rational functions.
Hence, our additive Ore--Sato theorem is not just a straight-forward
analog of its multiplicative predecessor, but is based on a significantly
different proof strategy.

\section{WZ-forms and structure of WZ-pairs}\label{sec:wzpair}

The goal of this section is to introduce some notions that will help us to describe the proofs in the later sections more concisely.
\begin{definition}[(Pairwise) shift-invariant]\label{def:shiftinvariant}
  A rational function $ f\in K(\mathbf{x}) $ is called \emph{shift-invariant}
  if there exists a nonzero integer vector $ \mathbf{v}\in \bZ^n $ such that
  $ f(\mathbf{x}+\mathbf{v})=f(\mathbf{x}) $.
  It is called \emph{pairwise shift-invariant} if for all $i,j \in \{1,\ldots,n\}$,
  there exist $ s,t\in \bZ $, not both zero, such that $ \si_i^{s}(f)=\si_j^{t}(f) $.
\end{definition}
\begin{definition}[Integer-linearity]\label{def:integerlinear}
  An irreducible polynomial $ p\in K[\mathbf{x}] $ is called
  \emph{integer-linear} over~$K$ if there exist a univariate
  polynomial $ P\in K[z]$ and a nonzero integer vector
  $ \mathbf{v}\in \bZ^n $ such that
  \[
    p(\mathbf{x})=P(\mathbf{v} \cdot \mathbf{x}).
  \]
  Without loss of generality, we can assume that $ \gcd(v_1,\ldots,v_n)=1 $,
  because a common factor can be extracted and absorbed by~$P$.
  Such a vector $\mathbf{v}$ is called the \emph{integer-linear type} of~$p$.
  We say that $f \in K(\mathbf{x})$ is integer-linear of type $\mathbf{v}$
  if all the irreducible factors of its numerator and its denominator are
  of the common integer-linear type~$\mathbf{v}$.
\end{definition}
There is an efficient algorithm for the computation of the integer-linear
decomposition of multivariate polynomials~\cite{Huang2019}, which will be used
for computing additive decompositions in Section~\ref{sec:implementation}. The
next lemma reveals the equivalence between the pairwise shift-invariance and
the integer-linearity of a rational function.
\newcommand\citePropSevenAP{\cite[Proposition~7]{AbramovPetkovsek2002a}}
\begin{lemma}[\citePropSevenAP]\label{lemma:pairwiseinteger}
	A rational function $ f\in K(\mathbf{x}) $ is pairwise shift-invariant if and only if there exist a nonzero integer vector $ \mathbf{v} \in \bZ^n $ and a univariate rational function $ r\in K(z) $ such that
    \[f(\mathbf{x})=r(\mathbf{v} \cdot \mathbf{x}),\]
	i.e., $f$ is integer-linear of type $\mathbf{v}$.
\end{lemma}
Given the integer-linear type of~$f$, one can easily see that $f$ is pairwise shift-invariant. In contrast, the opposite direction of Lemma~\ref{lemma:pairwiseinteger} is not that obvious. However, it follows, by using an inductive argument, from the bivariate case that is illustrated in the following remark.
\begin{remark}\label{rmk:bivariateintegerlinear}
  Let $f\in K(x,y)$ satisfy $ \si_x^s\si_y^t(f)=f $ with $ s, t \in \bZ $
  not both zero. If $ s=0 $, then $ f $ is free of $ y $, which implies that
  $ f $ is integer-linear of type~$(1,0)$. Similarly, if $t=0 $, then $ f $ is
  integer-linear of type~$(0,1)$. If both of them are nonzero, then $ f$ is
  integer-linear of type~$(\bar{t},\bar{s})$, where $ \bar{t}=t/\gcd(s,t) $
  and $ \bar{s}=s/\gcd(s,t) $.
\end{remark}
According to Definition~\ref{def:integerlinear}, an element in $ K $ can be viewed as having any integer-linear type. But for a non-constant rational function whose factors are of the same integer-linear type, its type is unique. Such a type remains unchanged under addition and under application of shift operators.

We now introduce two special kinds of WZ-forms, namely exact WZ-forms and
uniform WZ-forms, which will play an important role in describing the
structure of general WZ-forms (see Theorem~\ref{thm:addOreSato}).
\begin{definition}[Exact WZ-form]
  A WZ-form $(f_1,\ldots,f_n)$ with respect to $ (\De_1,\ldots,\De_n) $ is
  said to be \emph{exact} if there exists $ g\in K(\mathbf{x}) $ such that
  $ f_i=\De_{i}(g)$ for all $i\in \{1,\ldots,n\} $.
\end{definition}
\begin{definition}[Uniform WZ-form]\label{def:uniformWZ-formdefinition}
  A WZ-form $(f_1,\ldots,f_n)$ with respect to $ (\De_1,\ldots,\De_n) $ is
  called a \emph{uniform WZ-form} if there exists an integer vector
  $\mathbf{v}$ such that each $f_i$ is integer-linear of type~$\mathbf{v}$.
\end{definition}
\begin{remark}\label{rmk:exactuniform}
  A WZ-form can be both exact and uniform: for example,
  $\bigl(\De_x(\frac{1}{x+y}),\De_y(\frac{1}{x+y})\bigr)$ is an exact WZ-pair
  where each component is integer-linear of type~$(1,1)$.
\end{remark}
In the remaining part of this section, we recall the structure
theorem~\cite{Chen2019} on WZ-pairs that is described in terms of exact and
cyclic pairs, see Theorem~\ref{thm:structureWZ-pairs}.
\begin{definition}[Cyclic operator]\label{def:cyclicoperator}
  Let $ G=\left\langle \si_1,\dots,\si_n \right\rangle$ be the free abelian
  group generated by the shift operators $\si_1,\dots,\si_n$.
  For any $ m\in \bZ $ and $ \tau \in G  $, define
  \[
    \frac{\tau^{m}-1}{\tau -1}:=
    \begin{cases}
      1+\tau+\dots+\tau^{m-1}, & \text{if } m > 0;\\
      0, &\text{if } m=0;\\
      -(\tau^{m}+\dots+\tau^{-1}), & \text{if } m <0.
    \end{cases}
  \]
\end{definition}
\begin{definition}[Cyclic pair]
  A WZ-pair $(f,g)$ w.r.t.\ $(\De_x,\De_y)$ is called a \emph{cyclic pair}
  if there exists $h\in K(x,y)$ that satisfies $\si_x^s(h)=\si_y^t(h)$ for some $s,t \in \bZ$, not both zero, such that
	\[
    f=\dfrac{\si_y^t-1}{\si_y-1}\bullet h \quad \text{and} \quad g=\dfrac{\si_x^s-1}{\si_x-1}\bullet h.
    \]
\end{definition}
Note that any cyclic pair is a uniform WZ-pair by Remark~\ref{rmk:bivariateintegerlinear}. The following theorem shows that each WZ-pair can be decomposed into one exact WZ-pair plus several cyclic pairs.
\newcommand\citeThmWZpairs{\cite[Theorem~3]{Chen2019}}
\begin{theorem}[Structure of WZ-pairs, \citeThmWZpairs]\label{thm:structureWZ-pairs}
  Any WZ-pair can be decomposed into one exact WZ-pair plus several cyclic WZ-pairs.
\end{theorem}
When it comes to a multivariate generalization of
Theorem~\ref{thm:structureWZ-pairs}, cyclic pairs will be replaced by uniform
WZ-forms, see Theorem~\ref{thm:decompthm}. For this purpose, we define orbital decompositions and orbital residues of rational functions in the next section.

\section{Orbital decompositions and orbital residues}\label{sec:orbitaldecomp}

In this section, we recall the notion of orbital decomposition of a rational
function, which was first used for studying the existence problem of
telescopers~\cite{Chen2016}, and propose a modified definition of discrete
residues, which were originally introduced in~\cite{ChenSinger2014}
with polynomial and elliptic analogs in~\cite{Hou2015, Singer2018}.

\begin{definition}[Shift-equivalence]\label{def:shiftequivalence}
  Let $ F $ be a subgroup of $\langle\si_1,\ldots,\si_n\rangle$.
  For $a,b\in K(\mathbf{x})$, we say that $a$ and $b$ are
  \emph{$F$-equivalent}, denoted by $a\sim_{F} b$, if there exists
  $\tau \in F  $ with $ \tau(a)=b $. We call the set
  \[
    [a]_{F}:=\left\lbrace \tau(a) \ \vert \ \tau \in F\right\rbrace
  \]
  the \emph{$F$-orbit of $a$}. Note that $ a\sim_{F} b $ implies $ [a]_{F}=[b]_{F} $.
\end{definition}

\begin{example}\label{ex:shiftequivalence}
    Let $b=4x+6y+5z$ and let $F$ be the subgroup $\langle\si_x,\si_y\rangle$
    of $G=\langle\si_x,\si_y,\si_z\rangle$. Then $b$ and $b+1$ are $G$-equivalent
    because $\tau(b)=b+1$ for $\tau=\si_x^{-1}\si_z\in G$. In contrast, $b$ and $b+1$ are not $F$-equivalent.
\end{example}

The orbital decomposition of a rational function $ f=P/Q \in K(\mathbf{x}) $
depends on the variable~$x_1$ and a subgroup~$F$. In order to define it, we
first focus on its denominator as a polynomial in~$x_1$, that
is, $Q\in K(\widehat{\mathbf{x}})[x_1]$ with $\xh:=x_2,\ldots,x_n$.  The
first step consists in factoring the polynomial~$Q$ completely over
$K(\xh)$. We sort all of its irreducible factors into distinct $F$-orbits as
follows:
\[
  Q=c\cdot\prod_{i=1}^{I}\prod_{j=1}^{J}\prod_{\tau\in \Lambda_{i,j}}\tau(b_i^j),
\]
where $c\in K(\xh)$, $ \Lambda_{i,j}$ are finite subsets of~$F$, and the $b_i\in
K(\xh)[x_1]$ are monic irreducible polynomials in distinct
$F$-orbits. Note that this factorization is unique up to the choice of the representative $b_i$ in each $F$-orbit.
Moreover, we impose on the sets $\Lambda_{i,j}$ the condition that
$ \tau(b_i)\neq \tau'(b_i) $ for $ \tau,\tau'\in \Lambda_{i,j} $ with $
\tau\neq \tau' $. In the second step, we compute the unique irreducible
partial fraction decomposition of~$f$ with respect to the above factorization:
\begin{equation}\label{eq:orbitaldecomp}
	f=p+\sum_{i=1}^{I}\sum_{j=1}^{J}\sum_{\tau\in \Lambda_{i,j}}\frac{a_{i,j,\tau}}{\tau(b_i^j)},
\end{equation}	
where $ p, a_{i,j,\tau} \in K(\xh)[x_1]$ with
$\deg_{x_1}(a_{i,j,\tau})<\deg_{x_1}(b_i)$ for all $ i,j,\tau $.
For a polynomial $b\in K(\xh)[x_1] $, a subgroup $F\leq G$, and $ j>0 $, we
define the following linear $K(\xh)$-subspace:
\begin{equation}\label{eq:orbitalsubspace}
	U_{b,j}^{F}:=\Span_{K(\xh)}\left\lbrace \frac{a}{\tau(b^j)} \ \middle|\ \tau\in F,\ a\in K(\mathbf{\widehat{x}})[x_1],\ \deg_{x_1}(a)<\deg_{x_1}(b)\right\rbrace .
\end{equation}
Note that in Equation~\eqref{eq:orbitaldecomp}, each sum
$\sum_{\tau}\frac{a_{i,j,\tau}}{\tau(b_i^j)}$ lies in the
corresponding subspace~$U^F_{b_i,j}$.
Since the decomposition~\eqref{eq:orbitaldecomp} exists for any $f\in K(\x)$, and since the orbits $[b]_F$ do not overlap, we obtain the following direct sum decomposition:
\begin{equation}\label{eq:directsum}
    K(\x)= K(\xh)[x_1]\oplus \Biggl(\bigoplus_{j>0}\ \bigoplus_{[b]_F}\ U^F_{b,j}\Biggr)  ,
\end{equation}
where $[b]_F$ runs over all orbits in $K(\xh)[x_1]/{\sim_F} $.
Such a direct sum decomposition is called~\cite{Chen2016} \emph{the orbital decomposition of $ K(\x) $ with respect to the variable $ x_1 $ and the group $ F $}.

According to the definition of $U^F_{b,j}$, it is easy to check that this linear subspace is closed under the application of any operator in $K(\xh)[F]$, that is, any operator of the form $\sum_{\tau\in F}c_\tau \tau$ with $c_\tau \in K(\xh)$.
The following lemma is a direct generalization of~\cite[Lemma~5.1]{Chen2016}.

\begin{lemma}\label{lemma:invariantspace}
	If $ f\in U_{b,j}^F $ and $ \theta\in K(\xh)[F] $, then $ \theta(f)\in U_{b,j}^F $.
\end{lemma}
\begin{theorem}\label{thm:decompkeepprop}
Let $f = p+\sum_{i=1}^{I}\sum_{j=1}^{J} f_{i,j}$ with $p\in K(\xh)[x_1]$ and $f_{i,j}\in U_{b_i,j}^F$ be an orbital decomposition of~$f$ with respect to~$x_1$ and~$F$, and let $\theta_1,\theta_2 \in K(\xh)[F]$. For $g\in K(\x)$, we have $\theta_1(f)=\theta_2(g)$
if and only if there exist $q\in K(\xh)[x_1]$ and $g_{i,j}\in U_{b_i,j}^F$ such that
$\theta_1(p)=\theta_2(q)$ and $\theta_1(f_{i,j})=\theta_2(g_{i,j})$ for all~$i,j$.
\end{theorem}
\begin{proof}
  The sufficiency is due to the linearity of the operators $\theta_1,\theta_2\in K(\xh)[F]$.
	For the necessity, suppose $ g=q+\sum_{i=1}^{I}\sum_{j=1}^{J} g_{i,j} $, where $q\in K(\xh)[x_1]$ and $g_{i,j}\in U_{b_i,j}^F$ for each~$i,j$.
    By Lemma~\ref{lemma:invariantspace}, the orbital decomposition of
	$ \theta_1(f) $ with respect to $ x_1$ and $ F $ is
	\[
	\theta_1(f)=\theta_1(p)+\sum_{i=1}^{I}\sum_{j=1}^{J} \theta_1(f_{i,j}).
	\]
	Similarly, we get
	\[
	\theta_2(g)=\theta_2(q)+\sum_{i=1}^{I}\sum_{j=1}^{J} \theta_2(g_{i,j}) .
	\]
	By the uniqueness of the direct sum decomposition~\eqref{eq:directsum}, we have $ \theta_1(p)=\theta_2(q) $ and $ \theta_1(f_{i,j})=\theta_2({g_{i,j}}) $ for each $ i,j $.
\end{proof}

For $ f\in K(\mathbf{x}) $, we say that $ f $ is \emph{$ \si_i $-summable} if there exists $g\in K(\x)$ such that $ f=\De_i(g) $. Let $ (f_1,\ldots,f_n) $ be a WZ-form w.r.t. $ (\De_1,\dots,\De_n) $. Then $ \De_i(f_1) $ is $ \si_1 $-summable, because we have $\De_i(f_1)=\De_1(f_i)$. The first step in our proof of Theorem~\ref{thm:addOreSato} is to decompose $ f_1 $ and to find the shift-invariance of each part.

Next, for the definition of orbital residues, let us look at
the orbital decomposition of $f\in K(\x)$ with respect to~$x_1$
and the subgroup $F=\langle\si_1\rangle$. In this case,
the decomposition~\eqref{eq:orbitaldecomp} can be written as
\begin{equation}\label{eq:parfrac}
  f= p+\sum_{i=1}^{I}\sum_{j=1}^{J}\sum_{\ell=0}^L\frac{a_{i,j,\ell}}{\si_1^\ell(d_i^j)},
\end{equation}
where the $d_i$ are irreducible polynomials in distinct $\langle\si_1\rangle$-orbits.

\begin{definition}[Orbital residue]\label{def:polynomialresidue}
  Let $f$ be given in the form~\eqref{eq:parfrac}, let $d\in K(\xh)[x_1]$ be
  irreducible, and let $j\in\{1,\dots,J\}$. If there exists $i\in\{1,\dots,I\}$ such that
  $d_i\in[d]_{\langle\si_1\rangle}$ (by the properties of the orbital
  decomposition, such $i$ is uniquely determined), then the \emph{orbital residue}
  of $f$ at $d$ of multiplicity~$j$, denoted by $\res_{\si_1}(f,d,j)$,
  is defined to be the $\langle\si_1\rangle$-orbit $[r]_{\langle \si_1 \rangle}$ with
  \[
    r := \sum_{\ell=0}^L\si_1^{-\ell}(a_{i,j,\ell}).
  \]
  If no such $i$ exists, we define $\res_{\si_1}(f,d,j)=0$. If it is clear
  from the context, we will abbreviate $[r]_{\langle\si_1\rangle}$ by~$[r]$.
\end{definition}
Note that the definition of orbital residue does not depend on the
representation~\eqref{eq:orbitaldecomp} of~$f$: if instead of $d_i$
some other representative of $[d_i]_{\langle\si_1\rangle}$ is used,
at the cost of changing the range of~$\ell$, then also the polynomial~$r$ in
Definition~\ref{def:polynomialresidue} changes, but it will stay in the
same $\langle\si_1\rangle$-orbit. This is the reason why the residue
is defined to be an orbit, instead of a single polynomial. Similarly,
we have $\res_{\si_1}(f,d,j)=\res_{\si_1}(f,d',j)$ whenever
$d\sim_{\left\langle \si_1\right\rangle }d'$.
\begin{example}\label{ex:orbitalresidue}
  Let $b=4x+6y+5z$ as in Example~\ref{ex:shiftequivalence} and let
  \[
    f = \frac{x}{b^2}+\frac{x+y}{(b+1)^2}+\frac{2x}{(b-3)^2}+\frac{2x+3}{(b+3)^2}.
  \]
  We observe that $b+1=\si_x(b-3)$ and that $b, b-3, b+3$ are in distinct
  $\left\langle \si_x\right\rangle$-orbits.
  By Definition~\ref{def:polynomialresidue}, we have
  \begin{equation*}
    \res_{\si_x}(f,b,2) = [x], \quad
    \res_{\si_x}(f,b-3,2) = [3x+y-1], \quad
    \res_{\si_x}(f,b+3,2) = [2x+3].
  \end{equation*}
\end{example}

\section{Additive decompositions of WZ-forms}\label{sec:decompofWZforms}

Exact and uniform WZ-forms are special kinds of WZ-forms. Conversely, the
following theorem shows that these two forms are the only basic building
blocks of all possible WZ-forms. This section is dedicated to proving the
following theorem, which is a multivariate generalization of
Theorem~\ref{thm:structureWZ-pairs}.
\begin{theorem}\label{thm:decompthm}
  Any WZ-form can be decomposed into one exact WZ-form plus several uniform WZ-forms.
\end{theorem}
First we recall the notion of isotropy group, which was introduced by
Sato~\cite{Sato1990} in order to prove the classical Ore--Sato theorem.
\begin{definition}[Isotropy group]\label{def:isotropygroup}
	Let $ p\in K[\mathbf{x}] $. The set
	\[G_p=\{\tau \in G \ \vert \ \tau(p)=p\}\]
	is a subgroup of $ G $, called the \emph{isotropy group} of $ p $ in $ G$.
\end{definition}
\begin{example}\label{ex:isotropy}
    Let $b=4x+6y+5z$ and $G=\langle\si_x,\si_y,\si_z\rangle$ be as in Example~\ref{ex:shiftequivalence}, and let $c=3y+2z$. The isotropy group of~$b$
    consists of all monomials $\si_x^i\si_y^j\si_z^k$ that satisfy $4i+6j+5k=0$, i.e.,
    $G_b=\bigl\langle\si_x\si_y\si_z^{-2},\si_y^5\si_z^{-6}\bigr\rangle$.
    Similarly, one finds
    $G_c=\bigl\langle\si_x,\si_y^2\si_z^{-3}\bigr\rangle$ and
    $G_{b\cdot c}=G_b\cap G_c=\bigl\langle\si_x^3\si_y^8\si_z^{-12}\bigr\rangle$.
\end{example}
Definition~\ref{def:isotropygroup} can directly be extended to rational functions.
The next lemma shows that shift-equivalent elements have the same isotropy group.
\begin{lemma}\label{lemma:isotropygroupinvariant}
	Let $ f,g\in K(\mathbf{x}) $. If $ f\sim_{G}g $, then $ G_f=G_g $.
\end{lemma}
\begin{proof}
  Let $ \si \in G $ such that $ f=\si(g) $. For $ \tau\in G_g $ we have
  $ \tau(g)=g $. Applying $\si$ to both sides of the equation yields
  $\si\bigl(\tau(g)\bigr)=\si(g)$. Since $ \si $ and $ \tau $ commute, we have
  $\tau\bigl(\si(g)\bigr)=\si(g)$, i.e., $\tau(f)=f$. Thus $ \tau \in G_f $,
  which implies that $ G_g\subseteq G_f $. Since $ \si^{-1}\in G $ such that
  $ g=\si^{-1}(f) $, we similarly have $ G_f\subseteq G_g $. Hence $G_f=G_g$.
\end{proof}
We recall a crucial lemma that led to the structure theorem of WZ-pairs.
Here it will be used to conduct the induction step in the proof of
Theorem~\ref{thm:decompthm}.
\newcommand\citeLemmaWZpairs{\cite[Lemma~6]{Chen2019}}
\begin{lemma}[\citeLemmaWZpairs]\label{lemma:cyclicproof}
	Let $ f\in K(x,y) $ be a rational function of the form
	\[f=\frac{a_0}{b^m}+\frac{a_1}{\si_y(b^m)}+\cdots+\frac{a_n}{\si_y^{n}(b^m)},\]
	where $ m,n\in \bN$ with $m>0$, $ a_0,\ldots,a_n,b \in K(y)[x] $ with $ a_n\neq 0 $. Moreover, we assume that $ \deg(a_i)<\deg(b) $, $b$ is irreducible and monic, and that $ \si_y^{i}(b)\not\sim_{\left\langle \si_x\right\rangle }\si_y^{j}(b)  $ for all $ i,j\in \{0,\ldots,n\} $ with $ i\neq j $. If for some $g\in K(x,y)$ we have $ \De_y(f)=\De_x(g) $, then there exists $ t\in\bZ $ such that $ \si_y^{n+1}(a_0)=\si_x^{t}(a_0) $, $ \si_y^{n+1}(b)=\si_x^{t}(b) $, and $ a_\ell=\si_y^{\ell}(a_0)$ for all $\ell \in \{0,\ldots,n\}$. Furthermore,
	for some $g_0\in K(y)$ we get
	\[
	f=\frac{\si_y^{n+1}-1}{\si_y-1}\bullet \frac{a_0}{b^m} \quad\text{and}\quad
	g=\frac{\si_x^{t}-1}{\si_x-1}\bullet \frac{a_0}{b^m} +g_0.
	\]
\end{lemma}
According to Remark~\ref{rmk:bivariateintegerlinear}, the bivariate
function~$f$ in Lemma~\ref{lemma:cyclicproof} is of a certain integer-linear
type. We will use this lemma to reduce the problem from the multivariate
case to the bivariate case, see the proof of Lemma~\ref{lemma:subspaceWZ-forms} below.

Recall that $G=\left\langle \si_1,\ldots,\si_n \right\rangle $ and
$\xh=x_2,\ldots,x_n$. Let $\omega=(f_1,\ldots,f_n)\in K(\mathbf{x})^n$ be a
WZ-form w.r.t.~$(\De_1,\ldots,\De_n)$. Applying the orbital
decomposition~\eqref{eq:orbitaldecomp} with respect to~$x_1$ and~$G$ to~$f_1$
yields
\begin{equation}\label{eq:orbitaldecompf1}
f_1=p+\sum_{i=1}^{I}\sum_{j=1}^{J}\sum_{\tau\in \Lambda_{i,j}}\frac{a_{i,j,\tau}}{\tau(b_i^j)},
\end{equation}
where for all $ i,j,\tau $ we have $ p, a_{i,j,\tau} \in K(\xh)[x_1]$ with $\deg_{x_1}(a_{i,j,\tau})<\deg_{x_1}(b_i)$ and $\Lambda_{i,j}\subset G$.
The following reduction formula is crucial in Abramov's algorithm for rational summation~\cite{Abramov1971,Abramov1975}.
\begin{fact}
For all $a,u\in K[\x]$ with $u\neq 0$ and automorphism $ \phi  $ of $ K(\mathbf{x}) $, we have
\begin{equation}\label{eq:reductionformula}
	\frac{a}{\phi^{m}(u)}=\phi(g)-g+\frac{\phi^{-m}(a)}{u},
\end{equation}
where
\begin{equation}\label{eq:reductiong}
	g=
	\begin{cases}
		\displaystyle\sum_{i=0}^{m-1}\frac{\phi^{i-m}(a)}{\phi^{i}(u)}, & \text{if } m \geq 0;\\[0.5em]
		\displaystyle -\sum_{i=m}^{-1}\frac{\phi^{i-m}(a)}{\phi^{i}(u)}, & \text{if } m < 0.
	\end{cases}
\end{equation}
\end{fact}

Let $ E:=\left\langle \si_2,\ldots,\si_{n}\right\rangle  $. Then each $ \tau \in G $ can be written as $ \si_1^{m}\lambda$ for some $ m\in \bZ $ and $\lambda \in E$. By taking $\phi= \si_1$ and $u=\lambda(b)$ in Formula~\eqref{eq:reductionformula}, we get
\begin{equation}\label{eq:reductionsi1}
 \frac{a}{\tau(b)}=\frac{a}{\si_1^m(u)}=\De_{1}(g)+\frac{\si_1^{-m}(a)}{u}
 =\De_{1}(g)+\frac{\si_1^{-m}(a)}{\lambda(b)},
\end{equation}
for some $g\in K(\x)$ of the form~\eqref{eq:reductiong}.
Applying the above reduction~\eqref{eq:reductionsi1} to each summand $ a_{i,j,\tau}/\tau(b_i^j) $ in Equation~\eqref{eq:orbitaldecompf1} yields
\begin{equation}\label{eq:abramovdecomp}
	f_1=\De_1(g_0)+\sum_{i=1}^{I}\sum_{j=1}^{J} \widetilde{f}_{1,i,j}
        \quad\text{with}\quad
        \widetilde{f}_{1,i,j} = \sum_{\lambda\in \widetilde{\Lambda}_{i,j}}
        \frac{\widetilde{a}_{i,j,\lambda}}{\lambda(b_i^j)},
\end{equation}
where $ g_0 \in K(\mathbf{x}) $, $ \widetilde{\Lambda}_{i,j}\subseteq E $, and $ \lambda(b_i)\not\sim_{\left\langle \si_1\right\rangle }\lambda'(b_i) $ whenever $ \lambda,\lambda'$ are two distinct elements from $\widetilde{\Lambda}_{i,j} $. Since the shift operators~$\sigma_1^{-m}$ preserve the degrees of the polynomials~$a_{i,j,\lambda}$, we have for all $i,j$ that $\widetilde{f}_{1,i,j}\in U_{b_i,j}^G$.
In fact,
\[
[\widetilde{a}_{i,j,\lambda}] = \res_{\si_1}\bigl(f_1,\lambda(b_i),j\bigr).
\]
We give an illustrative example to show how we can immediately obtain the orbital residue
via the reduction~\eqref{eq:abramovdecomp}. Note that the result is the same as specified
in Definition~\ref{def:polynomialresidue}.
\begin{example}[Continuing Example~\ref{ex:orbitalresidue}]
	Rewrite $ f $ as
	\[
	f=\frac{x}{b^2}+\frac{x+y}{\si_x^{-1}\si_z(b^2)}+\frac{2x}{\si_x\si_y^{-2}\si_z(b^2)}+\frac{2x+3}{\si_x^{-3}\si_z^3(b^2)}.
	\]
	First we get rid of the operator $ \si_x $ among all the denominators,
	\begin{align*}
	f={}&\De_{x}\biggl(-\frac{x+y}{\si_x^{-1}\si_z(b^2)}+\frac{2x-2}{\si_y^{-2}\si_z(b^2)}-\frac{2x+3}{\si_x^{-3}\si_z^3(b^2)}-\frac{2x+5}{\si_x^{-2}\si_z^3(b^2)}-\frac{2x+7}{\si_x^{-1}\si_z^3(b^2)}\biggr)\\
	&+\frac{x}{b^2}+\frac{x+y+1}{\si_z(b^2)}+\frac{2x-2}{\si_y^{-2}\si_z(b^2)}+\frac{2x+9}{\si_z^3(b^2)}.
	\end{align*}
	Note that $ \si_y^{-2}\si_z(b^2)=\si_x^{-3}\si_z(b^2) $, so we continue the reduction as follows:
	\[\frac{2x-2}{\si_y^{-2}\si_z(b^2)}=\De_x\biggl(-\frac{2x-2}{\si_x^{-3}\si_z(b^2)}-\frac{2x}{\si_x^{-2}\si_z(b^2)}-\frac{2x+2}{\si_x^{-1}\si_z(b^2)}\biggr)+\frac{2x+4}{\si_z(b^2)}.\]
	 Hence
	 \[
	 f= \De_x(g)+\frac{x}{b^2}+\frac{3x+y+5}{\si_z(b^2)}+\frac{2x+9}{\si_z^3(b^2)},
	 \]
	 for some $ g\in K(\x) $. We observe that $\bigl\{b^2, \si_z(b^2),\si_z^3(b^2)\bigr\} =
	\bigl\{b^2, (b+5)^2, (b+15)^2 \bigr\}$ are pairwise
	$ \left\langle \si_x \right\rangle  $-inequivalent, hence the reduction is done. We have
	\[
	\res_{\si_x}(f,b,2)= [x], \
	\res_{\si_x}\bigl(f,\si_z(b),2\bigr)=[3x+y+5],\
	\res_{\si_x}\bigl(f,\si_z^3(b),2\bigr)=[2x+9].
	\]
\end{example}

Using the $g_0$ that was obtained by Abramov's reduction~\eqref{eq:reductionformula},
we define an exact WZ-form $\omega_0:= \bigl(\De_1(g_0),\ldots,\De_n(g_0)\bigr)$,
which we remove from the given WZ-form~$\omega$. To this end, we let
$\widetilde{f}_i:=f_i-\De_i(g_0)$ and observe that $\bigl(\widetilde{f}_1,\ldots,\widetilde{f}_n\bigr)$ is still a WZ-form, which implies that for each $k\in \{2,\ldots,n\}$, $\De_{k}(\widetilde{f}_1)$
is $ \si_{1} $-summable. Note that $ \sum_{i=1}^{I}\sum_{j=1}^{J}\widetilde{f}_{1,i,j} $ is the orbital decomposition of $\widetilde{f}_1 $ with respect to $ x_1 $ and $ G $. By Theorem~\ref{thm:decompkeepprop}, for each $ i,j$, we have $\De_k\bigl(\widetilde{f}_{1,i,j}\bigr) $ is $ \si_1 $-summable. Then we can focus on each orbital component of $\widetilde{f}_1$ in a linear $K(\xh)$-subspace $ U_{b,m}^G $.

\begin{remark}\label{rmk:pairwise}
  We claim that $a\in K(\x)\setminus K(\xh)$ is pairwise shift-invariant if
  and only if for each $ k\in \{2,\ldots,n\} $, there exist $L_k, N_k\in \bZ$
  with $L_k \neq 0$, such that $ \si_{k}^{L_k}(a)=\si_{1}^{N_k} (a)$. The
  necessity follows from Definition~\ref{def:shiftinvariant}. For the
  sufficiency, we combine for any $k,s\in \{2,\ldots,n\}$ the
  $N_s$-fold application of $ \si_{k}^{L_k}(a)=\si_{1}^{N_k} (a)$
  with the $N_k$-fold application of  $ \si_{s}^{L_s}(a)=\si_{1}^{N_s} (a)$ to obtain
  \[
    \si_k^{L_kN_s}(a)=\si_1^{N_kN_s}(a)=\si_{s}^{L_sN_k}(a).
  \]
  If $ N_k=N_s=0 $, then $ a $ is free of $ x_k$ and $x_s $ which implies that
  $\si_k^{1}(a)=\si_s^{1}(a) $.
\end{remark}	
\begin{lemma}\label{lemma:subspaceWZ-forms}
	Let $ f_1= \sum_{\lambda\in \Lambda}{a_{\lambda}}/{\lambda(b^m)}\in U_{b,m}^G $ with $ \Lambda\subset E $ and the $ \lambda(b)$ being in distinct $\left\langle \si_1 \right\rangle  $-orbits. If $ \De_{k}(f_1) $ is $ \si_{1} $-summable for each $ k\in \{2,\ldots,n\} $, then all of the $ a_{\lambda} $ and $ b $ are integer-linear of the same type.
\end{lemma}
\begin{proof}
	By Remark~\ref{rmk:pairwise} and Lemma~\ref{lemma:pairwiseinteger}, it is sufficient to show that for each $ k\in \{2,\ldots,n\} $, there exist $ L_k, N_k \in \bZ  $  with $ L_k $ nonzero such that $ \si_{k}^{L_k}(b)=\si_{1}^{N_k} (b)$ and $ \si_{k}^{L_k}(a_{\lambda})=\si_{1}^{N_k} (a_{\lambda})$ for all $ \lambda\in \Lambda$. Let $ E_k:=\left\langle \si_2,\ldots,\si_{k-1},\si_{k+1},\ldots,\si_n\right\rangle $.
	For each $ \lambda\in \Lambda \subset E $, there exist $ t_{\lambda}\in \bZ, \ \eta_{\lambda}\in E_k  $ such that $ \lambda=\si_k^{t_\lambda}\eta_{\lambda} $, and therefore
	\[f_1=\sum_{\lambda \in \Lambda }\frac{a_{\lambda}}{\si_k^{t_\lambda}\eta_{\lambda}(b^m)}.\]
	By applying the reduction formula~\eqref{eq:reductionformula} once again, we can rewrite $ f_1 $ in the form
	\begin{equation}\label{eq:rewritef1}
	f_1=\De_1(f_{1,k})+\sum_{\eta\in \Lambda_k}\sum_{\ell=0}^{T_{\eta}}\frac{\widetilde{a}_{\eta,\ell}}{\si_{k}^{\ell}\eta(b^m)},
	\end{equation}
	where $\Lambda_k\subset E_k$, $ \eta(b)\not\sim_{\left\langle \si_1,\si_k\right\rangle } \eta'(b)$ if $ \eta\neq \eta' $, $ \si_k^{\ell}(b)\not\sim_{\left\langle \si_1 \right\rangle }\si_k^{\ell'}(b)$ if $\ell\neq \ell' $, and $ \widetilde{a}_{\eta,T_\eta}\neq 0 $ for each~$ \eta  $. Furthermore, we assume that this representation is such that $T_\eta \geq0 $ is as small as possible.
	Note that $ \sum_{\ell=0}^{T_{\eta}} {\widetilde{a}_{\eta,\ell}}/{\si_k^{\ell}\eta(b^m)}\in U_{\eta(b),m}^{ \left\langle \si_1,\si_k\right\rangle }$. Recall that by our assumption $ \De_k(f_1) $ is $ \si_1 $-summable. Then by Theorem~\ref{thm:decompkeepprop}, we have that $ \De_k\Bigl(\sum_{\ell=0}^{T_{\eta}} {\widetilde{a}_{\eta,\ell}}/{\si_k^{\ell}\eta(b^m)}\Bigr) $ is $ \si_1 $-summable for each~$\eta$.  Now Lemma~\ref{lemma:cyclicproof} implies that there exist integers $ S_{\eta} $ such that
\begin{gather}
  \si_k^{T_{\eta}+1}\bigl(\eta(b)\bigr) = \si_1^{S_{\eta}}\bigl(\eta(b)\bigr),
  \label{eq:shiftinvarianceofb}\\
  \si_k^{T_{\eta}+1}(\widetilde{a}_{\eta,0}) = \si_1^{S_{\eta}}(\widetilde{a}_{\eta,0}),
  \label{eq:shiftinvarianceofa}\\
  \widetilde{a}_{\eta,\ell} = \si_k^{\ell}(\widetilde{a}_{\eta,0}),\quad
  \forall \ell\in \{0,\ldots,T_{\eta}\}.
  \label{eq:shiftstepofa}
\end{gather}
Applying $\eta^{-1}$ to both sides of Equation~\eqref{eq:shiftinvarianceofb} yields
$\si_k^{T_{\eta}+1}(b)=\si_1^{S_{\eta}}(b)$ since $ G $ is commutative.
Since the $ \si_k^\ell(b) $ are in distinct $ \left\langle \si_1\right\rangle$-orbits,
we have $ T_{\eta}=T_{\eta'}$ and  $S_{\eta}=S_{\eta'}$ for any two
$\eta, \eta' \in \Lambda_k$. Let $ L_k:=T_{\eta}+1 $ and $ N_k:=S_\eta $, then $ L_k $
is the minimal positive integer such that
$ \si_k^{L_k}(b)\sim_{\left\langle \si_1\right\rangle } b $ and
$\si_k^{L_k}(b)=\si_1^{N_k}(b) $. According to Equations~\eqref{eq:shiftinvarianceofa}
and~\eqref{eq:shiftstepofa}, we have
$\si_k^{L_k}(\widetilde{a}_{\eta,\ell})=\si_1^{N_k}(\widetilde{a}_{\eta,\ell})$
for each~$\eta$ and~$\ell$. We observe that
\[
  \res_{\si_{1}}\bigl(f_1,\lambda(b),m\bigr)=[a_{\lambda}]
  \quad \text{and} \quad
  \res_{\si_1}\bigl(f_1,\si_k^{\ell}\eta(b),m\bigr)=[\widetilde{a}_{\eta,\ell}].
\]
	For each $ \lambda\in \Lambda $, there exists a unique pair $ (\eta,\ell) $ where $  \eta\in \Lambda_k,\ell\in \{0,\ldots,T_{\eta}\}$ such that $ \lambda(b)\sim_{\left\langle \si_1 \right\rangle }  \si_k^\ell\eta(b)$. By Definition~\ref{def:polynomialresidue} we have $a_{\lambda}\sim_{\left\langle \si_1 \right\rangle }\widetilde{a}_{\eta,\ell}   $.
	Now Lemma~\ref{lemma:isotropygroupinvariant} implies that $\si_k^{L_k}(a_\lambda)=\si_1^{N_k}(a_\lambda)$.
\end{proof}

We are now ready to give the proof of Theorem~\ref{thm:decompthm}.
\begin{proof}
  We proceed by induction on~$n$. For the base case $n=1$, the theorem follows
  from the fact that any univariate rational function is a uniform WZ-form.
  For $n\geq 2$ suppose that the theorem holds for any WZ-forms in $n-1$ variables.
  As in Lemma~\ref{lemma:subspaceWZ-forms}, we focus on each component of
  the orbital decomposition of~$f_1$ and rewrite it as in~\eqref{eq:rewritef1}.
  Next we use the cyclic operator to describe $ f_1 $ in a more precise way as
  \[
    f_1=\De_1(f_{1,k})+\frac{\si_k^{L_k}-1}{\si_k-1}\bullet
    \sum_{\eta\in \Lambda_k}\frac{\widetilde{a}_{\eta,0}}{\eta(b^m)}.
  \]
  Suppose that $L_k,N_k\in \bZ$ with $L_k\neq0$ such that
  \[
    \si_k^{L_k}\Biggl(\frac{\widetilde{a}_{\eta,0}}{\eta(b^m)}\Biggr) =
    \si_1^{N_k}\Biggl(\frac{\widetilde{a}_{\eta,0}}{\eta(b^m)}\Biggr) .
  \]
  For each $ k\in \{2,\ldots,n\} $, let
  \[
    f_k'= \De_k(f_{1,k})+\frac{\si_1^{N_k}-1}{\si_1-1} \bullet
    \sum_{\eta\in \Lambda_k}\frac{\widetilde{a}_{\eta,0}}{\eta(b^m)}.
  \]
  Then one can easily check that $ \De_k(f_1)=\De_1(f_k') $ with $ f_k' $ and $ f_1 $ being integer-linear of the same type. For $ k,\ell \in \{2,\ldots,n\} $ with $ k\neq \ell  $, we have
	$	\De_k(f_1)=\De_1(f_k') $ and $\De_\ell(f_1)=\De_1(f_\ell')$, from which it follows that
	\[
	\De_\ell\De_1(f_k')=\De_\ell\De_k(f_1)=\De_k\De_1(f_\ell').
	\]

Thus $ \De_1\bigl(\De_\ell(f_k')-\De_k(f_\ell')\bigr)=0 $, i.e., $ \De_\ell(f_k')-\De_k(f_\ell')\in K(\xh) $. By construction, we have $ f_{1,k}\in U_{b,m}^G $ and $ f_2',\ldots,f_n' \in U_{b,m}^G$. By Lemma~\ref{lemma:invariantspace}, also $\De_\ell(f_k')-\De_k(f_\ell')$ is an element of $ U_{b,m}^G$. According to the definition of $  U_{b,m}^G $ in~(\ref{eq:orbitalsubspace}), one has
\[
  U_{b,m}^G \cap K(\xh)=\{0\}.
\]
Thus $ \De_\ell(f_k')-\De_k(f_\ell')=0 $.
By Definition~\ref{def:uniformWZ-formdefinition}, $ (f_1,f_2',\ldots,f_n') $ is a uniform WZ-form in $U_{b,m}^G$, denoted by $\omega_{i,j}$ for some $i,j$.

In conclusion, from the orbital decomposition of~$f_1$, we can obtain a
WZ-form $(f_1,f_2',\ldots,f_n')$ which is one exact WZ-form $\omega_0$ plus
several uniform WZ-forms~$\omega_{i,j}$. Note that there may remain a WZ-form
of the form $ (0,f_2-f_2',\ldots,f_n-f_n')$. From the compatibility
conditions~\eqref{eq:compatiblehyperarithmetic}, we have for each
$ k\in \{2,\ldots,n\} $ that $ \De_1(f_k-f_k')=\De_k(0)=0 $, so
$ f_k-f_k' \in K(\xh)$. Hence the remaining form can be viewed as an
$(n-1)$-variable WZ-form w.r.t.~$(\De_2,\ldots,\De_n)$. By the induction
hypothesis, the proof is completed.
\end{proof}
Note that this decomposition is not unique, because of two aspects.
When a WZ-form is both exact and uniform, see Remark~\ref{rmk:exactuniform},
we choose to put it into the exact part, which minimizes the uniform part.
Second, the final result depends on the operators in~$G$ that are chosen
for the orbital decomposition.
Next we give an example to illustrate how the decomposition works.
\begin{example}\label{ex:tripleWZ}
  Let $ \omega=(f,g,h)\in K(x,y,z)^3 $ be a WZ-form with
  \begin{align*}
    f &= \sum_{\ell=0}^{3}\frac{1}{4x+6y+5z+\ell},\\
    g &= \sum_{\ell=0}^{5}\frac{1}{4x+6y+5z+\ell}+\sum_{\ell=0}^{2}\frac{1}{3y+2z+\ell},\\
    h &= \sum_{\ell=0}^{4}\frac{1}{4x+6y+5z+\ell}+\sum_{\ell=0}^{1}\frac{1}{3y+2z+\ell}.
  \end{align*}
  It is easy to check that $(f,g,h)$ satisfy the following compatibility conditions:
  \[
    \big\lbrace\De_y(f)=\De_x(g), \ \De_z(f)=\De_x(h),  \ \De_z(g)=\De_y(h)\big\rbrace.
  \]
  Let $b=4x+6y+5z$ be the same as in Example~\ref{ex:orbitalresidue}, while
  the rational function~$f$ here is different. In terms of~$b$ it can be
  written as
  \[
    f=\frac{1}{b}+\frac{1}{\si_x^{-1}\si_z(b)}+
    \frac{1}{\si_x^{-1}\si_y(b)}+\frac{1}{\si_x^{-3}\si_z^{3}(b)},
  \]
  but note that this representation is not unique.
  Similarly, let $c=3y+2z$ and rewrite
  \[
    \sum_{\ell=0}^{2}\frac{1}{3y+2z+\ell}=
    \frac{1}{c}+\frac{1}{\si_{y}^{-1}\si_z^{2}(c)}+\frac{1}{\si_z(c)}.
  \]
  Then we can decompose~$\omega$ into an exact WZ-form plus two uniform WZ-forms:
	\begin{align*}
		f&=\De_x(a+\bar{a}) + \biggl(\De_x(a_2)+ \frac{\si_y^2-1}{\si_y-1}\cdot\frac{\si_z^2-1}{\si_z-1}\bullet \dfrac{1}{b} \biggr) +
		\frac{\si_y^0-1}{\si_y-1}\cdot\frac{\si_z^3-1}{\si_z-1}\bullet \frac{1}{c}\\
		&=\De_x(a+\bar{a})+\biggl(\De_x(a_3)+\frac{\si_z^4-1}{\si_z-1}\cdot \frac{\si_y-1}{\si_y-1}\bullet \dfrac{1}{b}\biggr)+\frac{\si_z^3-1}{\si_z-1}\cdot\frac{\si_y^0-1}{\si_y-1}\bullet \frac{1}{c},\\
		g&=\De_y(a+\bar{a})+\biggl( \De_y(a_2)+\frac{\si_x^3-1}{\si_x-1}\cdot \frac{\si_z^2-1}{\si_z-1}\bullet \dfrac{1}{b}\biggr) +\frac{\si_x-1}{\si_x-1}\cdot \frac{\si_z^3-1}{\si_z-1}\bullet \frac{1}{c},\\
		h&=\De_z(a+\bar{a})+\biggl(\De_z(a_3)+\frac{\si_x^5-1}{\si_x-1}\cdot\frac{\si_y-1}{\si_y-1}\bullet \dfrac{1}{b}\biggr)+\frac{\si_x-1}{\si_x-1}\cdot\frac{\si_y^2-1}{\si_y-1}\bullet \frac{1}{c}.
	\end{align*}
	where
	\begin{align*}
	    &a=-\dfrac{1}{\si_x^{-1}\si_z(b)}-\dfrac{1}{\si_x^{-1}\si_y(b)}-\dfrac{1}{\si_x^{-3}\si_z^3(b)}-\dfrac{1}{\si_x^{-2}\si_z^3(b)}-\frac{1}{\si_x^{-1}\si_z^3(b)} ,\\
		&a_2=\frac{1}{\si_y\si_z(b)},\quad
		a_3=-\dfrac{1}{\si_x^{-1}\si_z^2(b)},\quad
		\bar{a}=-\frac{1}{\si_y^{-1}\si_z^{2}(c)}.
	\end{align*}
	As we can see, the first uniform WZ-form has type~$(4,6,5)$, while
        the second one has type~$(0,3,2)$.
\end{example}

\section{Structure of uniform WZ-forms}\label{sec:uniformWZ-forms}

Theorem~\ref{thm:decompthm} tells us how every WZ-form can be decomposed into
exact and uniform WZ-forms. While exact WZ-forms are easy to describe and to
construct, Definition~\ref{def:uniformWZ-formdefinition} only allows us to
check whether a given tuple is a uniform WZ-form, but this characterization is
not explicit enough to construct such forms. In this section, we use a
difference homomorphism in order to write a uniform WZ-form in terms of its
integer-linear type and a single univariate rational function. Then we finish
our proof of the additive Ore--Sato theorem.

Let $(A,\bsi) $ and $(A,\btau) $ be two difference
rings, where $ \bsi=(\si_1,\ldots,\si_n )$ and $
\btau=(\tau_1,\ldots,\tau_n) $. A homomorphism (resp.\ isomorphism)
$ \phi\colon A  \rightarrow  A $ is called a difference homomorphism
(resp.\ isomorphism) from $(A,\bsi) $ to  $(A,\btau) $
if $ \phi \circ \si_{i}=\tau_i \circ \phi $ for each $i \in \{1,\dots,n\}$.
In other words, for each~$i$ there is a commutative diagram:
\begin{center}
  \begin{tikzcd}
    A \arrow{r}{ \si_i } \arrow{d}{ \phi }& A \arrow{d}{ \phi }\\
    A  \arrow{r}{ \tau_i } & A
  \end{tikzcd}
\end{center}

\begin{lemma}\label{lemma:differencehomomorphism}
  Given a unimodular matrix $ \mathbf{D}\in \bZ^{n\times n} $, i.e.,
  $\mathbf{D}^{-1}\in \bZ^{n \times n}$, we define a ring isomorphism
  $\phi \colon K(\x) \to  K(\x)$ by $\phi(\x)=\mathbf{D}\cdot\x$.
  Furthermore, we let the $\sigma_i$ act on vectors as
  $\sigma_i(\x)=\x+\e_i$, where $\e_i$ denotes the $i$-th unit vector.
  If we define $\tau_i(\x)=\x+\mathbf{D}^{-1}\cdot\e_i$
  for all $i\in\{1,\dots,n\}$, then $ \phi $ is a difference isomorphism
  from $ \bigl(K(\mathbf{x}),\bsi\bigr) $ to
  $\bigl(K(\mathbf{x}),\btau\bigr)$.
\end{lemma}
\begin{proof}
  We have to check that $ \phi \circ \si_{i}=\tau_i \circ \phi $. For the
  left-hand side we get
  \[
    \phi(\si_{i}(f(\x))) = \phi(f(\x+\e_i)) = f(\mathbf{D}\cdot\x+\e_i),
  \]
  and the right-hand side gives
  \[
    \tau_i(\phi(f(\x))) = \tau_i\bigl(f(\mathbf{D}\cdot\x)\bigr) =
    f\bigl(\mathbf{D}\cdot(\x+\mathbf{D}^{-1}\cdot\e_i)\bigr) = f\bigl(\mathbf{D}\cdot\x+\e_i\bigr).
  \] This completes the proof. 
\end{proof}

Given $ f_1,\ldots,f_n\in K(\x) $ satisfying the compatibility
conditions~\eqref{eq:compatiblehyperarithmetic}, then
\cite[Theorem~2]{Bronstein2005} shows that there exists a difference ring extension $
\bigl(K(\x)[H],\bsi\bigr)$ of $ \bigl(K(\x),\bsi\bigr) $, where $ H $ is a
hyperarithmetic expression with certificates~$f_1,\ldots,f_n $. A difference
homomorphism from $ \bigl(K(\x),\bsi\bigr)$ to $ \bigl(K(\x),\btau\bigr)$
can naturally be extended to the corresponding difference ring extensions.

\newcommand\citePropNineAP{\cite[Proposition~9]{AbramovPetkovsek2002a}}
\begin{lemma}[\citePropNineAP]\label{lemma:integermatrix}
	For every integer vector $ \mathbf{v}=(v_1,\dots,v_n) $ there is an integer matrix $ \mathbf{D}\in \mathbb{Z}^{n\times n} $ with the first row $ \mathbf{v} $ and $ \det (\mathbf{D})=\gcd (v_1,\ldots,v_n) $.
\end{lemma}

Next we use such a matrix $\mathbf{D}$ to construct the difference homomorphism.
\begin{theorem}\label{thm:structureUniWZ-forms}
	Let $ \bigl(f_1(\mathbf{v}
	\cdot \mathbf{x}),\ldots,f_n(\mathbf{v}\cdot \mathbf{x})\bigr) $ be a uniform WZ-form of type~$ \mathbf{v} $, then there exist constants $ \mu_1,\ldots,\mu_n \in K $ and a univariate rational function $r \in K(z) $ such that for each $ i\in \{1,\ldots,n\} $,
	\[
	f_i(\mathbf{v}\cdot\mathbf{x})=\mu_{i}+\sideset{}{_0^{v_i}}\sum_{\ell}r(\mathbf{v} \cdot \mathbf{x}+\ell) .
	\]
\end{theorem}
\begin{proof}
	Let $ H(\mathbf{x}) $ be a hyperarithmetic expression, and let $f_1(
	\vv\cdot\mathbf{x}),\ldots,f_n(\vv\cdot\mathbf{x})$ be its  certificates. That is to say for each $ i $,
	\begin{equation}\label{eq:hyperarithmetic}
		\si_{i}\bigl(H(\mathbf{x})\bigr)=H(\mathbf{x}) + f_i(\mathbf{v} \cdot \mathbf{x}).
	\end{equation}
	Without loss of generality, we can assume that $ \gcd(v_1,\ldots,v_n)=1 $. By Lemma~\ref{lemma:integermatrix}, there exists an integer matrix $ \mathbf{D}=(d_{ij})\in \bZ^{n\times n} $ with $ \det (\mathbf{D})=1 $ whose first row is $ \mathbf{v} $.
	Let $ \phi\colon K(\mathbf{x}) \rightarrow K(\mathbf{x}) $ such that
	\[ \phi\bigl(f(\mathbf{x})\bigr)=f\bigl(\mathbf{D}^{-1}\cdot\mathbf{x}\bigr),\ \text{for all } f(\mathbf{x})\in K(\mathbf{x}). \]
	By Lemma~\ref{lemma:differencehomomorphism}, $ \phi $ is a difference isomorphism from $ \bigl(K(\mathbf{x})[H],\bsi\bigr) $ to $ \bigl(K(\mathbf{x})[H],\btau\bigr) $, where $ \tau_i(\x)=\x+\mathbf{D}\cdot\e_i$ for all $ i $ in $\{1,\ldots,n\} $.
	Applying the operator $ \phi $ to Equation~\eqref{eq:hyperarithmetic} yields
	\begin{align*}
		\phi\bigl(\si_{i}\bigl(H(\mathbf{x})\bigr)\bigr)&=
                \phi\bigl(H(\mathbf{x})\bigr) + \phi\bigl(f_i(\mathbf{v} \cdot \mathbf{x})\bigr),\\
		\tau_i \bigl(\phi\bigl(H(\mathbf{x})\bigr) & =\phi\bigl(H(\mathbf{x})\bigr) + f_i(x_1).
	\end{align*}
	Let $ H'(\mathbf{x})=\phi\bigl(H(\mathbf{x})\bigr) $, then it follows that  $ \tau_i\bigl(H'(\mathbf{x})\bigr)=H'(\mathbf{x}) + f_i(x_1)$.
	For any integer $ m>0$ and $ i\in \{1,\ldots,n\} $ we have
	\begin{align}
	  \tau_i^{m}\bigl(H'(\mathbf{x})\bigr)&= H'(\x) + \sum_{j=0}^{m-1}f_i\bigl(x_1+jd_{1i}\bigr)
          =: H'(\x) + f_{i,m}(x_1),\label{eq:tauim}\\
	  \tau_i^{-m}\bigl(H'(\mathbf{x})\bigr) &= H'(\x) - \sum_{j=1}^{m}f_i\bigl(x_1-jd_{1i}\bigr)
          =: H'(\x) + f_{i,-m}(x_1).\label{eq:tauimm}
	\end{align}
	Let $ \mathbf{D}^{-1}=:\bigl(\widetilde{d}_{ij}\bigr)_{n\times n} $, then for all $  i\in \{1,\ldots,n\} $ we can rewrite $\si_i$ in terms of $\tau_1,\dots,\tau_n$:
	\[
	  \si_{i} = \prod_{j=1}^{n}\tau_j^{\widetilde{d}_{ji}}.
	\]
	By applying \eqref{eq:tauim} and \eqref{eq:tauimm} repeatedly, we obtain $\si_{i}\bigl(H'(\x)\bigr) = H'(\x)+f_{i}'(x_1)$ for some univariate rational function~$f_i'$.
    That is to say, $\De_i\bigl(H'(\mathbf{x})\bigr)=f_{i}'(x_1)$. By the compatibility conditions~\eqref{eq:compatiblehyperarithmetic} we have that $\De_1\bigl(f'_k(x_1)\bigr) = \De_k\bigl(f'_1(x_1)\bigr) = 0$, and thus
    $ f_k' \in K$ for all $ k\in \{2,\ldots,n\} $.
	Then an easy induction shows that
	 \[H'(\mathbf{x})\simeq F'(x_1)+\sum_{k=2}^{n}f_{k}'x_k,\]
	 where $ F'(x_1) $ is a solution of the difference equation $ y(x_1+1)-y(x_1)=f_1'(x_1) $.
	Next, we can recover $ H(\x) $ as follows,
	\begin{align*}	
		H(\x)&= \phi^{-1}\bigl(H'(\x)\bigr)\\
		&=H'\bigl(\mathbf{D}\cdot\mathbf{x}\bigr)\\
		&\simeq F'(\mathbf{v}\cdot\x)+\sum_{k=2}^{n}f_{k}'\biggl(\sum_{i=1}^{n}d_{ki}x_i\biggr) \\
		&=F'(\mathbf{v}\cdot\x)+\sum_{i=1}^{n}\biggl(\sum_{k=2}^{n}f_{k}'d_{ki}\biggr)x_i,
	\end{align*}
	where $ F'(\mathbf{v}\cdot \x+1)-F'(\mathbf{v})=f_1'(\mathbf{v}\cdot \x) $.
	Write that $ \mu_{i}:=\sum_{k=2}^{n}f_{k}'d_{ki} $. Then for each $ i\in \{1,\ldots,n\} $,
	\[
		f_i(\mathbf{v}\cdot\mathbf{x})=\De_{i}\bigl(H(\mathbf{x})\bigr)=
		\begin{cases}
			\displaystyle \ \mu_{i}+\sum_{\ell=0}^{v_i-1}f_1'(\mathbf{v} \cdot \mathbf{x}+\ell),& \text{if }  v_i > 0;\\[0.4 em]
		\ 	\mu_i , & \text{if }  v_i=0; \\[0.4 em]
		\ 	\displaystyle \mu_{i}-\sum_{\ell=v_i}^{-1}f_1'(\mathbf{v} \cdot \mathbf{x}+\ell), & \text{if }   v_i <0.
		\end{cases}
	\]
	Finally we let the univariate rational function $ r$ be defined as $ f_1' $.
\end{proof}
Eventually we obtain Theorem~\ref{thm:addOreSato} by combining
Theorems~\ref{thm:decompthm} and~\ref{thm:structureUniWZ-forms}. Note that we
can disregard the $\mu_i$ since the constant tuple $(\mu_1,\ldots,\mu_n)$
itself can be viewed as an exact WZ-form. We show that any
hyperarithmetic expression can be described, up to conjugation, as a rational
function plus a $ K $-linear combination of polygamma functions.  First we
employ the partial fraction decomposition on the univariate function $ r $
over $ K $:
\[
  r(z)=\sum_s\sum_t\frac{\beta_{s,t}}{(z+\alpha_s)^t},
\]
where $\alpha_s, \beta_{s,t} \in K$ and $s, t\in \bN $, both with finite support.

According to the recurrence formula of polygamma functions~\cite[(5.15)]{NIST:DLMF}
\[
  \psi^{(t)}(z+1)-\psi^{(t)}(z)=\frac{(-1)^tt!}{z^{t+1}},\quad t=0,1,\ldots
\]
we have
\[
  \psi^{(t)}(z+\alpha_s+1)-\psi^{(t)}(z+\alpha_s)=\frac{(-1)^tt!}{(z+\alpha_s)^{t+1}}.
\]
Then the hyperarithmetic expression $ H' $ with certificates
\[
  \Bigl( \sideset{}{_0^{v_1}}\sum_{\ell}r(\mathbf{v} \cdot \mathbf{x}+\ell),\ldots,
  \sideset{}{_0^{v_n}}\sum_{\ell}r(\mathbf{v} \cdot \mathbf{x}+\ell)\Bigr)
\]
is conjugate to
\[
  \sum_s\sum_t \frac{\beta_{s,t+1}}{(-1)^t t!}\psi^{(t)}(\mathbf{v}\cdot \x+\alpha_s).
\]
\begin{corollary}
  Any hyperarithmetic expression is conjugate to
  \[
    a+\sum_{\vv\in V}\sum_s\sum_t \beta_{\vv,s,t}\psi^{(t)}(\vv\cdot \x+\alpha_{\vv,s}),
  \]
  where $a\in K(\x)$, $ V\subset\bZ^{n} $, $ s,t\in \bN $,
  and for each $\vv$, we have $ \beta_{\vv,s,t},\alpha_{\vv,s}\in K $.
\end{corollary}
\begin{example}
	Let $ H $ be a hyperarithmetic expression with certificates~$ (f,g,h) $ as in Example~\ref{ex:tripleWZ}. Then $ H $ is conjugate to
	$ \psi^{(0)}(4x+6y+5z)+\psi^{(0)}(3y+2z)$.
\end{example}

\section{An algorithm for the minimal decomposition of WZ-forms}\label{sec:implementation}

Now we will present an algorithm for computing additive representations of WZ-forms based on the recursive idea in the proof of Theorem~\ref{thm:addOreSato}. 
{Furthermore, we require that in such a representation the uniform WZ-forms are minimal,  which we call \lq \lq the minimal decomposition\rq \rq. Therefore with this algorithm we can decide the exactness of WZ-forms by checking whether the uniform part is zero or not.  

\begin{definition}[Additive representation]
  Given a WZ-form $\omega = (f_1,\ldots,f_n)$, there is a decomposition of the form
  \[
    \omega = \bigl(\De_1(a),\ldots,\De_n(a)\bigr)+
    \sum_{\mathbf{v}\in V}\Bigl(
    \sideset{}{_0^{v_1}}\sum_{\ell}r_{\mathbf{v}}(\mathbf{v} \cdot \mathbf{x}+\ell),\ldots,
    \sideset{}{_0^{v_n}}\sum_{\ell}r_{\mathbf{v}}(\mathbf{v} \cdot \mathbf{x}+\ell) \Bigr).
  \]
  We call the list $\bigl( a, V, \{r_{\mathbf{v}}\}_{\mathbf{v}\in V} \bigr)$ an
  \emph{additive representation} of~$ \omega $. 
This decomposition is called \lq \lq minimal \rq \rq if it has minimal degree in the denominator of each~$r_{\mathbf{v}}(z)$.
\end{definition}

Let $\omega = (f_1,\ldots,f_n)\in K(\x)^n$ be a WZ-form.
First, we apply Abramov's reduction~\cite{Abramov1971}
with respect to the variable $ x_1 $ to decompose $f_1$ into
\[f_1= \De_1(g_0)+\sum_{i=1}^{I}\sum_{j=1}^{J}\frac{a_{i,j}}{b_i^j},\]
where $ g_0\in K(\xh)[x_1] $, $ a_{i,j},b_i \in K[\xh][x_1] $ with
$ \deg_{x_1}(a_{i,j})<\deg_{x_1}(b_i) $,  and the $ b_i $ are in distinct
$ \langle\si_1 \rangle$-orbits.

By Lemma~\ref{lemma:subspaceWZ-forms}, each $ {a_{i,j}}/{b_i^j} $ is
integer-linear of some type~$ \mathbf{v}_i $.
In order to compute the type of each simple fraction in the above decomposition,
we are reduced to the following problem.
\begin{problem}[Integer-linear testing]
	Given a polynomial $p\in K[\x]$, decide whether there exist $u\in K[z]$ and $\mathbf{v}\in \bZ^n$ such that $p = u(\mathbf{v}\cdot \x)$.
\end{problem}
This problem can be solved by the algorithm \texttt{IntegerLinearDecomp}~\cite{Huang2019}.
Applying it to the numerator and the denominator of each simple fraction $a_{i, j}/b_i^j$ yields
\[\frac{a_{i,j}}{b_i^j}=u_{i, j}(\mathbf{v}_i\cdot \x),\]
where $u_{i, j}\in K(z)$ and $\mathbf{v}_i\in \bZ^n$ with the first entry $v_{i, 1}$ being nonzero.
By collecting the simple fractions of the same type, we obtain
\[f_1=\De_1(g_0)+\sum_{\mathbf{v}\in V}  u_{\mathbf{v}}(\mathbf{v}\cdot \x),\]
where $V \subset \bZ^n$ is a finite set and $u_{\mathbf{v}}\in K(z)$ for each ${\mathbf{v}}\in V$.
The next step is to write the rational function $u_{\mathbf{v}}$ into the form
\[
u_{\mathbf{v}}(z)=\sideset{}{_0^{v_1}}\sum_{\ell} r_{\mathbf{v}}(z +\ell),
\]
where $r_{\mathbf{v}} \in K(z)$. Note that $r_{\mathbf{v}}$ must be a rational solution of the difference equation
\[y(z+v_1) - y(z)=u_{\mathbf{v}}(z+1) - u_{\mathbf{v}}(z),\]
which can also be solved by Abramov's reduction.

Let $ \omega_0:= \bigl(\De_1(g_0),\ldots,\De_n(g_0)\bigr) $ and $ \omega_{\mathbf{v}}:=\bigl(f_{1,\mathbf{v}},\ldots,f_{n,\mathbf{v}}\bigr) $, where for each $ k\in \{1,\ldots,n\} $,
\[
f_{k,\mathbf{v}} := \sideset{}{_0^{v_k}}\sum_{\ell} r_{\mathbf{v}}(\mathbf{v}\cdot\x +\ell).
\]
Then $ \omega $ can be written as a summation of one exact WZ-form, several uniform WZ-forms and a ``degenerate'' WZ-form:
\[
\omega= \omega_0+\sum_{\mathbf{v}\in V} \omega_{\mathbf{v}}+\widetilde{\omega}.
\]
In order to get the minimal uniform WZ-forms in the sense that the denominator of each~$r_{\mathbf{v}}$ has the lowest possible degree in~$z$, we require the application of Abramov's reduction on~$r_{\mathbf{v}}$ again. The following example illustrates that this step is needed to minimize~$r_{\mathbf{v}}$.
	\begin{example}
		Let $ \omega=(f,g,h)\in K(x,y,z)^3 $ be a WZ-form with
		\begin{align*}
			f &= \Delta_x\biggl(\frac{1}{(4x+6y+5z+1)^2}\biggr)+\sum_{\ell=0}^{3}\frac{1}{4x+6y+5z+\ell},\\
			g &= \Delta_y\biggl(\frac{1}{(4x+6y+5z+1)^2}\biggr)+ \sum_{\ell=0}^{5}\frac{1}{4x+6y+5z+\ell},\\
			h &= \Delta_z\biggl(\frac{1}{(4x+6y+5z+1)^2}\biggr)+ \sum_{\ell=0}^{4}\frac{1}{4x+6y+5z+\ell}.
		\end{align*}
		Since the residue of Abramov's reduction is not unique, we may get the following decomposition of~$f$,
		\begin{align*}
			f &= \Delta_x\left( \frac{1}{(4x+6y+5z+1)^2} - \frac{1}{4x+6y+5z+1}\right)  \\
			& + \frac{1}{4x+6y+5z} + \frac{1}{4x+6y+5z+5}+ \frac{1}{4x+6y+5z+2} + \frac{1}{4x+6y+5z+3} .
		\end{align*}
		The univariate function in~$Z$ corresponding to the uniform WZ-form is
		\[r_{\mathbf{v}}(Z) = \frac{1}{Z} - \frac{1}{Z+1} + \frac{1}{Z+2}. \]
		Obviously we would anticipate it to be~${1}/{Z}$. This accident can happen because $4x+6y+5z+1$ and $4x+6y+5z+5$ are in the same $\sigma_x$-orbit, and Abramov's reduction on~$f$ cannot see which one will lead to the minimal~$r_{\mathbf{v}}$. However, we observe that the difference is the summable part of~$r_{\mathbf{v}}$ with respect to~$Z$, which finally can be removed by modifying the exact part of~$f$. In this case,
		\[
		\sum_{\ell=0}^{3} \left(\frac{1}{Z+2+\ell} - \frac{1}{Z+1+\ell}\right) = \frac{1}{Z+5}-\frac{1}{Z+1}.
		\]
		After substituting $Z$ with~$4x+6y+5z$ on the right-hand side of the above equation, we obtain~$\Delta_x\bigl( \frac{1}{4x+6y+5z+1}\bigr)$. The modification is done by absorbing it into the previous exact part of~$f$.
	\end{example}
If $\widetilde{\omega}$ is nonzero, then we proceed with the induction step by repeating the above process for
$ \widetilde{\omega} $ which only involves $(n-1) $-variables.
The above process for computing additive representations of WZ-forms
is summarized in Algorithm~\ref{alg:WZFormDecomp} and is illustrated in
Example~\ref{ex:decomp}. Note that $\omega$ is exact if and only if the output is~$(a,\varnothing,\varnothing)$, which is equivalent to that the uniform part is zero. 
Our Maple code for implementing Algorithm~\ref{alg:WZFormDecomp} is available at
\[
\text{\url{http://www.mmrc.iss.ac.cn/~schen/AddOreSato.html}}
\]
We provide a worst-case complexity analysis of Algorithm~\ref{alg:WZFormDecomp} in terms of arithmetic operations in the base field.
Let $K[\mathbf{x}]_d$ denote the set of polynomials in $K[\mathbf{x}]$ whose degree in $x_i$ is no more than $d$ for each~$i = 1,2,\ldots,n$. Let $K(\mathbf{x})_d$ denote the set of rational functions in $K(\mathbf{x})$ with numerators and denominators in~$K[\mathbf{x}]_d$. The \lq\lq big Oh\rq\rq~notation $O$ is referred to the cost of an algorithm ``up to a constant factor" and the \lq\lq soft Oh\rq\rq~notation $\widetilde{O}$ may further neglect some logarithmic factors. We define the max-norm of the multivariate polynomial~$f = \sum_{i_1,\ldots,i_n} f_{i_1,\ldots,i_n} x_1^{i_1}\cdots x_n^{i_n} \in \mathbb{Q}[\mathbf{x}] $ as
\[ \| f \|_{\infty} = \max_{i_1,\ldots,i_n} |f_{i_1,\ldots,i_n}| .\]

We first recall some known facts on complexity estimates.
 The irreducible factorization of a univariate polynomial~$f \in \bZ[x]$ of degree~$d \geq 1$ with max-norm $\| f \|_{\infty} =A$ takes $\widetilde{O}(d^6\cdot(d+\log A))$ operations in $\bZ$ (see~\cite[Theorem 16.23]{BookMCA2013}).
The partial fraction decomposition of an element in $K(x)_d$ with a given factorization of its denominator takes no more than $\widetilde{O}(d)$ operations in~$K$ (see~\cite[Section 5.11]{BookMCA2013}).
Multi-point evaluation and interpolation in $K[\mathbf{x}]_d$ from the given values on $O(d^m)$ points which form an $m$-dimensional tensor product grid can be done with $\widetilde{O}(md^m)$ operations in $K$ (see~\cite[Theorem 1]{JorisEric2013}).
Let $f(\mathbf{x}) = \prod_{\ell=1}^{s} \bar{f}_\ell(\mathbf{x})$ be a multivariate factorization of~$f(\mathbf{x})\in \bZ[\mathbf{x}]$ over~$\mathbb{Q}$. Then Mahler's bound in~\cite{Mahler1962} using Gelfond's inequality  leads to
\[
	\prod_{\ell=1}^{s} \|\bar{f}_\ell\|_{\infty} \leq 2^{nd}(d+1)^{n/2} \|f\|_{\infty}.
\]
According to~\cite{Huang2019}, the multivariate integer-linear decomposition of $f \in \mathbb{Z}[\mathbf{x}]$ with~$\|f\|_{\infty} = A$ takes $(n + \log A +d)^{O(1)}$ word operations.

Let~$f_1,\ldots, f_n \in \mathbb{Z}(\mathbf{x})_d$ with their max-norms bounded by~$A \in \bN$.
We claim that the total cost of Algorithm~\ref{alg:WZFormDecomp} is $\widetilde{O}\bigl(n^{O(1)}d^{O(n)}\log A\bigr)$ operations in~$\mathbb{Q}$.
The cost of the algorithm before the recursive loop (line 30  of Algorithm~\ref{alg:WZFormDecomp}) is dominated by the irreducible partial fraction decomposition of $f_1 \in \mathbb{Q}(\mathbf{x})_d$ with respect to~$x_1$,
where all simple fractions are of degree in~$x_j, \ j=2,\ldots,n $ bounded by~$d^2$.
To obtain partial fraction decompositions with respect to $x_1$ over $\bQ(x_2, \ldots, x_n)$, we take the strategy of multi-point evaluation and interpolation, i.e., first evaluating the rational function for variables  $(x_2, \ldots, x_n)$
at many points, performing the univariate irreducible partial fraction decompositions with respect to $x_1$, and then recovering the desired decompositions over $\bQ(x_2, \ldots, x_n)$ by rational interpolation.
To be sufficient for recovering the final result,  we need to evaluate~$x_2,\ldots,x_n$ at~$O(d^{2(n-1)})$ many good points which takes~$\widetilde{O}\bigl((n-1)d^{2(n-1)}\bigr)$ operations in~$\mathbb{Q}$.
For each point, the univariate irreducible partial fraction decompositions takes~$\widetilde{O}(d^6 \cdot (d+ \log A))$ operations in~$\mathbb{Q}$,
which is the dominating cost of the irreducible factorization of the denominator. Finally, we need the rational interpolation to recover~$O(d)$ coefficients in~$\bQ(x_2,\ldots,x_n)$, where each coefficient costs $\widetilde{O}\bigl((n-1)d^{2(n-1)}\bigr)$ operations in~$\mathbb{Q}$. Thus before line 30, the total cost is
\[\widetilde{O}\bigl((n-1)d^{2(n-1)}\cdot d + d^{2(n-1)}\cdot d^{6}\cdot (d+\log A) \bigr).\]

From the special structure of WZ-forms, we observe that the degree bound will not be changed in the recursive steps,
and the max-norms of denominators appeared in irreducible partial fraction decompositions are uniformly bounded by $\widetilde{A} := 2^{nd}(d+1)^{n/2}A$. The exact parts during the computation will
always be expressed in a sparse representation. This idea was used to get a polynomial-time algorithm for univariate rational summation in~\cite{Gerhard2003}. Let~$T(n,d,A)$ denote the worst-case running cost of Algorithm~\ref{alg:WZFormDecomp}  in terms of operations in~$\mathbb{Q}$. According to the previous discussions, we have
\[T(n,d,A) = \widetilde{O}\left((n-1)d^{2n-1} + d^{2n+4}\cdot (d+ \log A)\right) + \widetilde{T}(n-1, d, \widetilde{A}),\]
where  $\widetilde{T}(n-1, d, \widetilde{A})$ satisfies the recursive formulae
\begin{gather*}
	\widetilde{T}(m-1,d, \widetilde{A}) = \widetilde{O}\bigl((m-2)d^{2m-3} + d^{2m+2}\cdot (d+ \log \widetilde{A})\bigr) + \widetilde{T}(m-2, d, \widetilde{A})
\end{gather*}
for all $m = 3,\ldots, n$.
By solving the recurrence relation, we conclude that~$T(n,d,A)$ is $\widetilde{O}\bigl(n^{O(1)}d^{O(n)}\log A\bigr)$.


\begin{algorithm}
	\caption{WZ-form decomposition algorithm}  \label{alg:WZFormDecomp}
	\textbf{Function:} \texttt{WZFormDecomp}$\bigl(
	(f_1,\ldots,f_n),\x,Z\bigr)$\\
	\textbf{Input:} WZ-form $ (f_1,\ldots,f_n)\in K(\x)^n $, $\x=(x_1,\dots,x_n)$, and
	a new variable~$ Z $\\
	\textbf{Output:} Its additive representation: $ \bigl(a, V, R = \{r_{\mathbf{v}}\}_{\mathbf{v}\in V}\bigr)$
	\begin{algorithmic}
		\If{$ f_1=0 $}
		\State $( {a}, {V}, {R}) \gets \texttt{WZFormDecomp}\bigl(
		(f_2,\ldots,f_n),(x_2,\ldots,x_n),Z\bigr)$
		\For{$ {\mathbf{v}}=(v_2,\ldots,v_n)$ in $V$}
		\State $ {\mathbf{v}}\gets (0,v_2,\ldots,v_n) $
		\EndFor
		\State\Return $ ( {a}, {V}, {R}) $
		\EndIf
		
		\State  Call \texttt{AbramovReduction}:
		$f_1= \De_1(g_0)+\sum_{i=1}^{I}\sum_{j=1}^{J}{a_{i,j}}/{b_i^j}$, { $a \gets g_0$}
		\If{$ n=1 $}
		\State\Return $\bigl(g_0, \bigl((1)\bigr), \bigl(f_1-\De_1(g_0)\bigr) \bigr)$
		\EndIf
		
		\For{$ 1\leq i \leq I $}
		\State Call \texttt{IntegerLinearDecomp}: $b_i=q_i(\vw_i\cdot \x) $
		with $q_i\in K[Z]$
		\EndFor
		\State $ V \gets (\vv_1,\ldots,\vv_m)$ with $\{\vv_1,\ldots,\vv_m\}=\{\vw_1,\dots,\vw_I\}$
		and $\vv_i\neq\vv_j$ for $i\neq j$
		\For{$ 1 \leq k \leq m $}
		\State $ u_k \gets 0 $
		\For{$1 \leq i \leq I  $}
		\If{the integer-linear type of $b_i$ is $\vv_k = (v_{k,1},\ldots,v_{k,n}) $}
		\For{$ 1 \leq j \leq J $}
		\State Perform the substitution $\vv_k\cdot\x\to Z$ in $a_{i,j}$
		so that $a_{i,j}\in K[Z]$
		\State $u_k \gets u_k+ a_{i,j}/q_i^j$
		\EndFor
		\EndIf
		\EndFor
		\State Call \texttt{AbramovReduction}: $\si_z(h_k)-h_k=u_{k}(v_{k,1}z+1)-u_{k}(v_{k,1}z)$
		\State $ r_k \gets h_k(1/v_{k,1} Z) $, { $r_k = \Delta_Z(g_k) + \widetilde{r}_k $, $a \gets a + g_k(\vv_k \cdot \x)$}
		\EndFor
		\State { $ R \gets (\widetilde{r}_1,\ldots,\widetilde{r}_m) $}
		\For{$ 2 \leq k \leq n $}
		\State $\displaystyle f_k' \gets f_k{- \Delta_k(a) - \sum_{i=1}^m \sideset{}{_0^{v_{i,k}}}\sum_\ell
			\widetilde{r}_i(\mathbf{v}_{i}\cdot \x + \ell)}$
		\EndFor
		\If{$ f_k' \neq 0 $ for some $ k $}
		\State $(a', V', R')\gets
		\texttt{WZFormDecomp}\bigl((f_2',\ldots,f_n'),(x_2,\ldots,x_n),Z\bigr)$
		\For{$\mathbf{v}'=(v_2,\ldots,v_n)$ in $V'$}
		\State $\mathbf{v}'\gets (0,v_2,\ldots,v_n) $
		\EndFor
		\State $ a \gets a+a' $, \ $ V \gets \texttt{Join}(V, V')$, \ $ R \gets \texttt{Join}(R, R') $
		\EndIf
		\State \Return $ (a,V,R) $
	\end{algorithmic}
\end{algorithm}

\begin{example}
  \label{ex:decomp}
  Let $\omega = (f,g,h) \in K(x,y,z)^3$ be a WZ-form with respect to
  $ (\De_x,\De_y,\De_z) $, specifically,
  \begin{align*}
    f &= \frac{xyz-y^2z-yz^2+yz-1}{x-y-z+1},\\
    g &= \frac{x^2z-xyz-xz^2+xy-y^2-yz-1}{x-y-z},\\
    h &= \frac{x^2y-xy^2-xyz+xz-yz-z^2-1}{x-y-z}.
  \end{align*}
  Employing Abramov's reduction on $ f $ yields
  \[f= \De_x(xyz)+\frac{1}{-x+y+z-1}.\]
  Then we record the following exact WZ-form as a part of $ \omega $:
  \[ \omega_0:=\bigl(\De_x(xyz),\De_y(xyz),\De_z(xyz)\bigr) .\]
  Obviously from the decomposition of $ f $ there is only one integer-linear type $ \mathbf{v}=(-1,1,1) $ and the corresponding univariate rational function is $ r_{\mathbf{v}}=1/Z $. It is easily checked that there is no summable part in~$r_{\mathbf{v}}$. Then a uniform WZ-form shows up as a part of~$ \omega $:
  \[
    \omega_{\mathbf{v}}=\Bigl(\frac{1}{-x+y+z-1},\frac{1}{-x+y+z},\frac{1}{-x+y+z}\Bigr) .
  \]
  Then we can update $ \omega$ by subtracting $ \omega_0$ and
  $ \omega_{\mathbf{v}} $ and obtain $\widetilde{\omega}=(0,y,z)$, which is
  equivalent to the WZ-pair $ (y,z) $ with respect to $ (\De_y,\De_z) $. By
  simple manipulations we can see that it is an exact WZ-pair:
  \[
    \Bigl( \De_y\bigl(\tfrac{1}{2}y^2+\tfrac{1}{2}z^2\bigr),
    \De_z\bigl(\tfrac{1}{2}y^2+\tfrac{1}{2}z^2\bigr)\Bigr) .
  \]
  Combining this exact WZ-form with the previous one we can update $ \omega_0 $ as:
  \[
    \omega_0=\Bigl(
    \De_x\bigl(xyz+\tfrac{1}{2}y^2+\tfrac{1}{2}z^2\bigr),
    \De_y\bigl(xyz+\tfrac{1}{2}y^2+\tfrac{1}{2}z^2\bigr),
    \De_z\bigl(xyz+\tfrac{1}{2}y^2+\tfrac{1}{2}z^2\bigr)
    \Bigr).
  \]
  Finally the decomposition works as $ \omega=\omega_0+\omega_{\mathbf{v}} $, i.e.,
  the additive representation of $ \omega $ is
  \[
    \Bigl(xyz+\tfrac{1}{2} y^2+\tfrac{1}{2} z^2,\left\lbrace(-1,1,1)\right\rbrace,\{1/Z\}\Bigr).
  \]
\end{example}

\section*{Acknowledgment}
We would like to thank Hui Huang for providing the Maple code of the integer-linear decomposition. We also thank Jing Guo and Hanqian Fang for helpful discussions. We are grateful for the comments of the anonymous referees that helped us to improve the paper significantly.



\bibliographystyle{elsarticle-harv}
\bibliography{AddOreSato}

\end{document}